\documentclass[12pt]{article}
\usepackage{amsmath, amsthm, amssymb}
\usepackage[utf8]{inputenc}
\usepackage{graphicx}
\usepackage[mathscr]{euscript}
\usepackage[english]{babel}
\usepackage[pdftex,dvipsnames]{xcolor}
\usepackage[noend]{algpseudocode}
\usepackage{algorithm}
\usepackage{algorithmicx}
\usepackage{hyperref}
\usepackage[a4paper, top=2.5cm,bottom=2.5cm,left=2.5cm,right=2.5cm, marginparwidth=1.75cm]{geometry}
\usepackage{subfig}
\usepackage{float}

\usepackage{ifxetex}
\ifxetex
\usepackage{fontspec}
\usepackage{xunicode}
\usepackage{xltxtra}
\usepackage{xecyr}
\usepackage{polyglossia}
\setdefaultlanguage{english}
\else
\usepackage[utf8]{inputenc}
\DeclareUnicodeCharacter{2217}{\Large{\textbf{!!!FIX ME!!!}}}
\usepackage[T2A]{fontenc}
\fi

\usepackage{fullpage}
\usepackage{eufrak}
\usepackage{color}
\usepackage{xcolor}
\usepackage{listings}
\usepackage{hyperref}
\usepackage{enumitem}
\usepackage{amsthm,amsmath,amsfonts,amssymb}

\usepackage{tikz}
\usetikzlibrary {positioning}
\usetikzlibrary{automata}
\usetikzlibrary{arrows,shapes}
\tikzstyle{vertex}=[circle,fill=black,minimum size=3pt,inner sep=0pt]
\tikzstyle{edge} = [draw,thick,-]

\newtheorem{theorem}{Theorem}[section]

\newtheorem{lemma}{Lemma}[section]

\newtheorem{cor}{Corollary}[section]

\newcommand{\cykalgo}{%
\begin{algorithm}[H]
\caption{CYK algorithm}
\label{alg:cyk}
\hspace*{\algorithmicindent} \textbf{Input:} Grammar $G = (N, \Sigma, P, S)$ in CNF, string $w = a_1 \ldots a_n, \, a_i \in \Sigma$ \\
\hspace*{\algorithmicindent} \textbf{Output:} True if $w \in \mathcal{L}(G)$, False otherwise

\begin{algorithmic}[1]
\If{$w = \epsilon$}
    \State \textbf{return} $S \rightarrow \epsilon \in P$    
\EndIf
\For{$i = 1$ to $n$}
    \State $T_{i-1, i} = \{A|A \rightarrow a_i \in P\}$
\EndFor
\For{$l = 2$ to $n$}
    \For{$j = 0$ to $n - l$}
        \State $j = i + l$
        \State $T_{i, j} = \emptyset$
        \For{$A \rightarrow BC \in P$}
            \For{$k = i + 1$ to $j - 1$}
                \If{$B \in T_{i, k}$ and $C \in T_{k, j}$}
                    \State $T_{i, j} = T_{i, j} \cup \{A\}$
                \EndIf
            \EndFor
        \EndFor
    \EndFor
\EndFor
\State \textbf{return} $S \in T_{0, n}$ 
\end{algorithmic}
\end{algorithm}
}

\newcommand{\cflalgo}{%
\begin{algorithm}[H]
\caption{CFL reachability algorithm}
\label{alg:cfl}
\hspace*{\algorithmicindent} \textbf{Input:} Grammar $G = (N, \Sigma, P, S)$ in CNF; directed labeled graph $D = (V, E, L), L \subset \Sigma$\\
\hspace*{\algorithmicindent} \textbf{Output:} Pairs of $S$-reachable vertices in $D$
\begin{algorithmic}[1]
\State $W = \emptyset$
\If{$S \rightarrow \epsilon \in P$}
    \For{$v \in V$}
        \State Insert edge $v \xrightarrow{S} v$ in $D$ and triple $(v, v, S)$ in $W$
    \EndFor
\EndIf
\For{$A \rightarrow a \in P$}
    \For{$v \xrightarrow{a} u \in E$}
        \State Insert edge $v \xrightarrow{A} u$ in $D$ and triple $(v, u, A)$ in $W$
    \EndFor
\EndFor
\While{$W \neq \emptyset$}
    \State Extract triple $(v, u, A)$ from $W$
    \For{$u \xrightarrow{B} w \in E$}
        \For{$C \rightarrow AB \in P$}
            \If{$v \xrightarrow{C} w \notin D$}
                \State Insert edge $v \xrightarrow{C} w$ in $D$ and triple $(v, w, C)$ in $W$
            \EndIf
        \EndFor
    \EndFor
    \For{$w \xrightarrow{B} v \in E$}
        \For{$C \rightarrow BA \in P$}
            \If{$w \xrightarrow{C} u \notin D$}
                \State Insert edge $w \xrightarrow{C} u$ in $D$ and triple $(w, u, C)$ in $W$
            \EndIf
        \EndFor
    \EndFor
\EndWhile
\State \textbf{return} $\{(u, v)| u \xrightarrow{S} v \in E\}$ 
\end{algorithmic}
\end{algorithm}
}

\let\origenumerate\enumerate
\let\origendenumerate\endenumerate
\renewenvironment{enumerate}{\origenumerate[topsep = 0pt, noitemsep]}{\origendenumerate}

\providecommand{\keywords}[1]
{
  \small	
  \textbf{Keywords:} #1
}

\usepackage[backend=bibtex]{biblatex}

\bibliography{article}

\begin{document}

\title{Fine-grained reductions around CFL-reachability}
\author{Aleksandra  Istomina \\ \small{st062510@student.spbu.ru}
\and
Semyon Grigorev \\ \small{s.v.grigoriev@spbu.ru}
\and
Ekaterina Shemetova \\ \small{katyacyfra@gmail.com} \\
\small{Saint-Petersburg University, Saint Petersburg, Russia} 
}
\date{}
\maketitle              
%


\vspace{-2em}

\begin{abstract}
In this paper we study the fine-grained complexity of the CFL reachability problem. We first present one of the existing algorithms for the problem and an overview of conditional lower bounds based on widely believed hypotheses. We then use the existing reduction techniques to obtain new conditional lower bounds on CFL reachability and related problems. We also devise a faster algorithm for the problem in case of bounded path lengths and a technique that may be useful in finding new conditional lower bounds.

\small{\keywords{CFL reachability, fine-grained complexity, conditional lower bounds}}
\end{abstract}

 \vspace{-1.7em}

\begin{section}{Introduction}
    Context-free language (CFL) reachability is a framework for graph analysis which was introduced by Thomas Reps~\cite{REPS1998701} and Mihalis Yannakakis~\cite{10.1145/298514.298576} and allows one to specify path constraints in terms of context-free languages. CFL reachability finds application in such fields of research as static code analysis (e.g. type-based flow analysis~\cite{10.1145/373243.360208} or points-to analysis~\cite{10.1145/1103845.1094817, 10.1145/1133255.1134027}), graph databases~\cite{10.1145/298514.298576}, bioinformatics~\cite{SubgraphQueriesbyContextfreeGrammars}.
	
	CFL reachability problem is formulated as follows. One is given context-free language $\mathcal{L}$ over finite alphabet $\Sigma$ and a directed graph $D=(V, E, L)$ which edges are labeled with the symbols from $\Sigma$. It is needed to decide whether there exist a path between a pair of vertices in $D$ on which labels on the edges form a word from $\mathcal{L}$. There exist two possible variations of a problem, where the question of defined above reachability is asked for a selected pair $(s, t)$ of vertices --- s-t reachability, or for all pairs of them -- all-pairs reachability. CFL reachability problem complexity is often measured by the number of vertices in the graph: $n = |V|$, while the grammar is supposed to be fixed. There are several cubic~\cite{10.1145/298514.298576, 10.1145/199448.199462} and slightly subcubic~\cite{10.1145/1328438.1328460} (with time $\mathcal{O}(n^{3} / poly(\log n))$\footnote{Every logarithm in the paper is a logarithm with base 2.}) algorithms for all-pairs CFL reachability, so the problem clearly lies in complexity class \textbf{P}. 
	
	Fine-grained complexity theory studies the exact degree of the polynomial expressing the complexity of the problems. In respect to CFL reachability the big open question is whether a truly subcubic (with time $\tilde{\mathcal{O}}(n^{3 - \epsilon})$\footnote{Here and throughout the paper $\tilde{\mathcal{O}}$ notation hides polylogarithmic factors,  e.g. $\tilde{\mathcal{O}}(n^{3 - \epsilon})~=~\mathcal{O}(n^{3 - \epsilon}~\cdot~poly(\log(n))$.}) algorithm exists. Another, equally interesting, question is whether the fine-grained complexity of s-t and all-pairs CFL reachability are the same.
	
	One of the ways to answer the question about an exact complexity of a problem is to find a lower bound. It is hard to create unconditional lower bound at most of the times. At the moment for the most of the problems there are known only trivial lower bounds based on the necessity to read all the input and produce all the output. As an example of a non-trivial lower bound one can see $\mathcal{O}(n \log n)$ lower bound for sorting algorithms based on comparisons, but this case is the exception rather than the rule.
	
	In fine-grained complexity the usual way to answer a question about an exact complexity of a problem is to create a conditional lower bound: the lower bound is true if some widely believed hypothesis is true. The most popular hypotheses in the field are stated for the SAT, 3SUM and APSP problems.
	
	Conditional lower bounds are achieved via a fine-grained reduction which shows how to convert an algorithm for problem $A$ into an algorithm for problem $B$. In this case fast enough algorithm for $A$ can lead to breakthrough algorithm for $B$ which is not believed to exist. 

\end{section}

\section{Problem statement}

The aim of this work is to analyze CFL reachability problem from the fine-grained complexity point of view. In order to achieve the aim, the following objectives were set.

\begin{enumerate}
    \item Study existing algorithms for CFL reachability and explore possible ways of their enhancement.
    \item Study existing conditional lower bounds on the problem and their techniques.
    \item Find new conditional lower bounds on CFL reachability (preferably based on SAT, 3SUM and APSP hypotheses) or present reasons of their non-existence.
\end{enumerate}

The rest of the paper is organised as follows. We introduce the necessary definitions and basic algorithms in the field in Sec.~\ref{sec:prelim}. After that in Sec.~\ref{sec:map} we structure the results in the area and present a map with connections between the problems. In Sec.~\ref{sec:lower} we present some lower bounds and argue on the possibility of some reductions. Sec.~\ref{sec:bounded_paths} is devoted to the faster algorithms for CFL reachability if path lengths are bounded by some value. In Sec.~\ref{sec:line_edges} we discuss a technique that may be useful in creating new reductions and its limitations. 

\begin{section}{Preliminaries}
\label{sec:prelim}

In this section we introduce the necessary definitions, discuss basic algorithms for the CFL reachability and recognition problems and give an example of a fine-grained reduction.

\subsection{Definitions}
\label{subsec:def}

\textbf{CFL Reachability}
	\emph{Context-free grammar} (CFG) is a four-tuple $G=(N, \Sigma, P, S)$, where:
	
	\begin{itemize}
	    \item $N$ is a set of nonterminals,
	    \item $\Sigma$ is a set of terminals,
	    \item $P$ is a set of productions of the followings form: $A \to \alpha$, $\alpha \in (N \cup \Sigma)^*$
	    \item and $S \in N$ is a starting nonterminal.
	\end{itemize}
	
	We call $|G| = |N| + |\Sigma| + |P|$ the \textit{size of the grammar}. 
	
	We write $\alpha A \beta \Rightarrow \alpha \gamma \beta$ where $\alpha, \gamma, \beta \in (\Sigma \cup N)^*$ if $A \rightarrow \gamma$ is a production, we call it a \textit{derivation step}. \textit{Derivation} is a composition of $i \ge 0$ derivation steps and is represented as $\Rightarrow^*$. Denote a context-free language of words derived from the starting nonterminal as $\mathcal{L}(G)$, i.e. $\mathcal{L}(G) = \{w \in \Sigma^*| S \Rightarrow^* w\}$.
	
	Grammar is said to be in \textit{Chomsky-Normal Form (CNF)} if every production is in one of the following forms:
	
	\begin{enumerate}
	    \item $A \rightarrow BC, B, C \in N$
	    \item $A \rightarrow a, a \in \Sigma$
	    \item $S \rightarrow \epsilon$, if $S$ is not used in right-hand side of the productions
	\end{enumerate}
	
	Every CFG can be transformed to CFG in CNF in polynomial time of its size~\cite{CHOMSKY1959137}. Further in some algorithms we will need the input grammar be in CNF and we assume that it can be done in $poly(|G|)$ time as preprocessing without affecting the time complexity of the algorithm.
	
	\emph{CFG recognition} problem is to decide whether $w \in \mathcal{L}(G)$ given a CFG $G$ and a string $w \in \Sigma^*$. This problem is closely related to \emph{CFG parsing problem} where we want a possible derivation sequence, if $w \in \mathcal{L}(G)$. It is known~\cite{10.5555/646233.682379} that CFG recognition is as hard as CFG parsing up to logarithmic factors. 
	
	Let $D = (V, E, L)$ be a directed graph with $n$ vertices which edges are labeled with symbols from $L \subseteq \Sigma$. Denote by $v \xrightarrow{a} u$ and edge from vertex $v$ to vertex $u$ in $D$ labeled with symbol $a$. We call a path from vertex $v$ to vertex $u$ an \textit{$A$-path} if concatenation of labels on that path is a word that can be derived from the nonterminal $A \in N$. We say that $v$ is \textit{$A$-reachable} from $u$ (or that $(u, v)$ is an $A$-reachable pairs of vertices) if there exists an $A$-path from $u$ to $v$. In case if $A = S$ is the starting nonterminal we may also say that $v$ is \textit{$\mathcal{L}$-reachable} from $u$.
	
	\emph{Context-free language (CFL) reachability} problem~\cite{REPS1998701} is to determine if there exists an $S$-path between some vertices. In \emph{single source/single target (s-t)} CFL reachability there is one pair of vertices $s, t \in V(D)$ for which we want to determine the existence of an $S$-path from $s$ to $t$. This variation of a problem is also called \emph{on-demand problem}, because it can be useful in case of many reachability queries. On the other hand in  \emph{all-pairs} CFL reachability we are asked about $S$-path existence between all possible pairs of vertices in a graph, all at once. 
	
	\emph{Dyck-$k$} reachability problem is a CFL reachability problem where $G$ defines a Dyck language on $k$ types of parentheses. The corresponding grammar is $G=(N, \Sigma, P, S)$, where:

	\begin{itemize}
	    \item $N = \{S\}$
	    \item $\Sigma = \{(_i, )_i\}, \forall i = 1, \ldots, k$
	    \item productions rules are $S \rightarrow \epsilon | SS | (_1 S )_1 | \ldots | (_k S )_k$, where $\epsilon$ is the empty string. 
	\end{itemize}
	
	Dyck language is known~\cite{ch-sch} to be the hardest among the context-free languages in a way that every other language can be seen as a combination of some Dyck-$k$ with a regular language. Moreover s-t reachability problem can be transformed~\cite{schepper2018complexity} into Dyck-2 s-t reachability problem in $\mathcal{O}(|E|)$ time if we suppose the size of the given grammar to be fixed (which is a traditional assumption).
	
	The PDA Emptiness problem is closely connected to CFL reachability. We begin with the definitions used in it.
	
	\textit{Pushdown Automaton (PDA)}~\cite{schepper2018complexity} is six $\mathcal{A} = (Q, \Sigma, \Gamma, \delta, q_0,Q_f)$ where:

    \begin{itemize}
        \item $Q$ is a finite set of states
        \item $\Sigma$ is a finite string alphabet
        \item $\Gamma$ is a finite stack alphabet
        \item $\delta \subseteq Q \times (\Sigma \cup \{ \epsilon \}) \times (\Gamma \cup \{ \epsilon \}) \rightarrow Q \times \Gamma^*$ is the finite transition function 
        \item $q_0$ is a start state
        \item $Q_f \subseteq Q$ is a set of final states
    \end{itemize}

    For some string $w$ the PDA reads it starting in state $q_0$ and traverses through the states using transition function. String $w$ is said to be accepted by $A$ if after reading $w$ PDA $\mathcal{A}$ stops in state $q$ that lies in $Q_f$. $\mathcal{L}(\mathcal{A})$ is the language of words accepted by PDA $\mathcal{A}$. \emph{PDA Emptiness} problem is the problem of determining by $\mathcal{A}$ is $\mathcal{L}(\mathcal{A})$ empty.

    Reachability in push-down systems is the problem analogous to CFL reachability problem in model checking. \textit{A push-down system (PDS)~\cite{hansen2021tight}} $\mathcal{P}$ is a triple $(Q, \Gamma, \delta)$, where:

 \begin{itemize}
     \item $Q, |Q| = n$ is a finite set of control states
     \item $\Gamma$ is a finite stack alphabet
     \item $\delta \subseteq Q \times \Gamma \times Q \times \Gamma^*$ is a transition relation.
 \end{itemize}   
 
 \textit{A configuration of $\mathcal{P}$} is a pair $(q, w) \in Q \times \Gamma^*$, where $q$ is a control state and $w$ is the stack word. \textit{Configuration graph} $\mathcal{G}_{\mathcal{P}}$ if a graph on configurations as vertices with edge relation defined as follows: edge $(q_1, w_1) \rightarrow (q_2, w_2)$ exists when $(q_1,\gamma,q_2,w) \in \delta$, where $w \in \Gamma^*$ and $\gamma \in \Gamma \cup \{ \epsilon \}$ are such that either:

\begin{enumerate}
    \item $w_1 = \gamma = \epsilon$ and $w_2 =w$, or
    \item $\gamma \neq \epsilon$ and there exists $w^{\prime} \in \Gamma^*$ such that $w_1 =\gamma w^{\prime}$ and $w_2 = ww^{\prime}$.
\end{enumerate}

\textit{The (state) reachability} problem for a PDS $\mathcal{P}$ asks, given two states $q_s,q_t \in Q$, to decide whether there exists a path $P: (q_s , \epsilon) \rightarrow^* (q_t , \epsilon)$ in $\mathcal{G}_{\mathcal{P}}$, where $\rightarrow^*$ denotes a path consisting of non-negative number of edges between the corresponding vertices. We analogously define \emph{all-pairs PDS reachability}. The problem is called \textit{sparse} if $| \delta | = \mathcal{O}(|Q|)$. We call \emph{the stack depth} the maximum size of the stack that the system $\mathcal{P}$ can have. Below we will always work with systems that have stack depth no more than some $k$, we call $k$ the \emph{upper bound on a stack depth}.
	
\textbf{Fine-grained reductions}
	For problems $P, Q$ and time bounds $t_P, t_Q$, a \emph{fine-grained reduction}~\cite{bringmann2019fine} from $(P, t_P)$ to $(Q, t_Q)$ is an algorithm that, given an instance $I$ of $P$, computes an instance $J$ of $Q$ such that: 
	
	\begin{itemize}
		\item $I$ is a YES-instance of $P$ if and only if $J$ is a YES-instance of $Q$,
		\item for any $\epsilon > 0$ there is a $\delta > 0$ such that $t_Q(|J|)^{1 - \epsilon} = \mathcal{O}(t_P (|I|)^{1 - \delta})$, 
		\item the running time of the reduction is $\mathcal{O}(t_P (|I|)^{1 - \gamma})$ for some $\gamma > 0$.
	\end{itemize}

\subsection{Existing problems and hypotheses}
	
    Here we give definitions for several problems that are connected with CFL recognition and reachability. 
	
\textbf{Popular problems}
	\emph{Boolean satisfiability problem (SAT, $k$-SAT)} is to determine if there exists an interpretation of variables that satisfies a given Boolean formula on $n$ variables written in $k$-conjunctive normal form, $k \ge 2$ ($k$-CNF). The hypotheses about SAT, that we are interested about, are:
	
	\begin{itemize}
	    \item[SETH~\cite{10.1006/jcss.2000.1727}] which proposes that there is no $\epsilon > 0$ such that $k$-SAT can be solved in time $\mathcal{O}(2^{(1 - \epsilon) n})$ for any $k$.
	    \item[NSETH~\cite{10.1145/2840728.2840746}] which proposes that there is no $\epsilon > 0$ such that $k$-SAT can be solved co-nondeterministically in time $\mathcal{O}(2^{(1 - \epsilon) n})$ for any $k$.
	\end{itemize}
	
	The \textit{3SUM} problem~\cite{williams2018some} is as follows: given a set $S$ of $n$ integers from $\{- n^c, \ldots, n^c\}, c > 0$, determine whether there are a triple that sums in zero, i.e. $a, b, c \in S$ such that $a + b + c = 0$. 3SUM hypothesis states that there exists no algorithm that runs in time $\mathcal{O}(n^{2 - \epsilon}), \epsilon > 0$ for this problem.
	
	In \textit{All-Pairs Shortest Path} (APSP) problem~\cite{williams2018some} one is given an $n$ node edge-weighted graph $U = (V, E)$. As in the previous problem all weights are integers from $\{- n^c, \ldots, n^c\}, c > 0$. It is supposed that there are no cycles of total negative weight in the graph. For each pair of vertices $v, u$ one needs to find the weighted distance between them, i.e. the minimum over all paths from $v$ to $u$ of the total weight sum of the edges of the path. Hypothesis states that no algorithm with time $\mathcal{O}(n^{3 - \epsilon}), \epsilon > 0$ can solve APSP problem. 
	
	In \emph{Boolean Matrix Multiplication (BMM)} problem~\cite{williams2018some} it is needed to calculate matrix product of the two given $n \times n$ matrices over boolean semiring. The fastest algorithm currently known~\cite{10.5555/3458064.3458096} solves the task in $\mathcal{O}(n^{2.372..})$ using algebraic methods. The minimal constant $c$ such that the BMM problem can be solved in $\mathcal{O}(n^c)$ is known as $\omega$, matrix multiplication exponent. Combinatorial BMM hypothesis~\cite{10.1145/505241.505242} states that there is no $\mathcal{O}(n^{3 - \epsilon}), \epsilon > 0$ combinatorial algorithm for BMM. Combinatorial algorithms are not explicitly defined but the idea behind the term is that the algorithm does not use algebraic tricks to get smaller degree in polynomial; also combinatorial algorithms are more often used in practice because algebraic tricks lead to big constants behind $\mathcal{O}$-notation. 
	
	\emph{Orthogonal Vectors (OV)} problem decides whether two sets $X, Y$ of $n$ boolean $d~=~\omega(\log n)$-dimensional vectors contain a pair $x \in X, y \in Y$ which dot product equals zero. Conjecture states that OV problem cannot be solved in $\mathcal{O}(n^{2 - \epsilon} \cdot poly(d)), \forall \epsilon > 0$ time. 
	
	Given an undirected graph $U$ on $n$ vertices the \emph{$k$-Clique} problem~\cite{abboud2018if} seeks the clique on $k$ vertices in $U$.  If $0 \leq F \leq \omega$ and $0 \leq C \leq 3$ are the smallest numbers such that $k$-Clique can be solved combinatorially in $\mathcal{O}(n^{\frac{Ck}{3}})$ time and in $\mathcal{O}(n^{\frac{Fk}{3}})$ time by any algorithm, for any constant $k \geq 1$, a conjecture~\cite{williams2018some} is that $C = 3$ and $F = \omega$.
	
	\emph{Triangle detection} problem~\cite{hansen2021tight} is a case of $k$-Clique problem with $k=3$, it asks whether an undirected graph $G = (V , E)$ contains three nodes $i, j, k \in V$ with $(i,j),(j,k),(k,i) \in E$. It can be solved in $\mathcal{O}(n^3)$ time by combinatorial algorithms, and in $O(n^{\omega})$ time in general. 
	
	\emph{Language Editing Distance (LED)} problem is about determining the minimum number of corrections: insertions, deletions, substitutions, that are needed to transform the given string $w \in \Sigma^*$ into the string $w'$ that lies in the given context-free language $\mathcal{L}(G)$ over the same alphabet $\Sigma$. This problem is known~\cite{10.1137/0201022} to be solvable in $\mathcal{O}(n^3)$ time, where $n$ is the length of the string and the grammar is considered to be fixed.
	
\textbf{Dynamic problems}
	The following problems are connecting CFL reachability with dynamic problems, where for one input there can be multiple queries or updates. 
	
	In the \emph{Online boolean Matrix-Vector multiplication (OMV)} problem~\cite{10.1145/2746539.2746609} we are given an $n \times n$ boolean matrix $M$, we receive $n$ boolean vectors $v_1, \ldots, v_n$ one at a time, and are required to output $Mv_i$ (over the boolean semiring) before seeing the vector $v_{i+1}$, for all $i$. It is conjectured that there is no algorithm with total time $\mathcal{O}(n^{3-\epsilon})$ for this problem, even with polynomial time to preprocess $M$.
	
	The incremental \emph{Dynamic Transitive Closure (DTC)}~\cite{Hanauer2020FasterFD} problem asks to maintain reachability information in a directed graph $D = (V, E)$ between arbitrary pairs of vertices under insertions of edges. Conditional lower bound on DTC follows from OMV hypothesis and reduction from it~\cite{10.1145/2746539.2746609}: there is no algorithm with total update time $\mathcal{O}((mn)^{1 - \epsilon})$ ($n = |V|, m = |E|$) even with $poly(n)$ time preprocessing of the input graph and $m^{\delta - \epsilon}$ time per query for any $\delta \in (0, 1/2]$ such that $m = \Theta(n^{1/(1-\delta)})$ under OMV hypothesis.
	
\textbf{Problems with lower bound based on SAT, 3SUM or APSP}
	Below are listed two problems that have conditional lower bounds under SAT, 3SUM or APSP hypotheses and thus reductions from them to CFL reachability problem will create conditional lower bounds based on widely believed conjecture.
	
	Our first problem of interest is \textit{AE-Mono$\Delta$} problem where we are presented with a graph $G = (V, E)$ on $n$ vertices in which every edge $e_{uv}$ is colored with color $c(e_{uv})$, where $c$ is color function. AE-Mono$\Delta$ problem asks for every edge $e_{uv}$ whether there exists a monochromatic triangle that includes this edge, i.e. there is vertex $w$ that $e_{uw}, e_{vw} \in E(G)$ and $c(e_{uv}) = c(e_{uw}) = c(e_{vw})$. This problem is known~\cite{10.1007/11786986_24, 10.1145/1798596.1798597} to be solvable in $n^{(3 + \omega) / 2}$ time and has~\cite{https://doi.org/10.4230/lipics.itcs.2020.53} $n^{2.5}$ conditional lower bound under 3SUM and APSP hypotheses (which is tight if $\omega = 2$).
	
	Second problem is really a pair of problems, both of them have~\cite{10.1145/2746539.2746594} $n^{3}$ conditional lower bound under SETH, 3SUM and APSP hypotheses, where $n$ is the number of vertices in the given graph. In the \textit{Triangle collection} problem one is given a node-colored graph, and the question is whether for all triples of distinct colors there exists a triangle in a graph which vertices have these colors. \textit{$\Delta$ Matching Triangles} problem have the same input but in this problem one is asked about existence of a triple of colors for which there are at least $\Delta$ triangles in a graph which vertices are colored in this triple.

\subsection{Basic algorithms and examples}
\label{subsec:basic}

In this section we present the basic algorithms for the defined CFL problems and a classic example of the fine-grained reduction.

\textbf{CFG recognition}

Cocke–Younger–Kasami (CYK) algorithm~\cite{10.5555/524279} is a cubic algorithm for the CFL recognition problem. Recall that in CFL recognition problem it is needed to decide whether the given string $w = a_1 \ldots a_n$ lies in the given context-free language $\mathcal{L}(G)$. CYK algorithm for each substring $a_{i + 1} \ldots a_j, 0 \le i < j \le n$ of $w$ builds set $T_{i, j}$:

$$T_{i, j} = \{A \in N|A \Rightarrow^* a_{i + 1} \ldots a_j\}$$

After all the sets are built what is left is only to check if $S \in T_{0, n}$, which by construction of the sets is equivalent to $S \Rightarrow^* w$  .

CYK algorithm needs as an input a string and a grammar $G$ in CNF. Recall from Sec.~\ref{subsec:def} that every grammar can be translated into CNF in $poly(|G|)$ time. Algorithm is based on the idea that if $A \Rightarrow^* w'$ then there exists a production $A \rightarrow BC$ and a partition of the string $w' = w'' \cdot w'''$ such that $B \Rightarrow^* w''$ and $C \Rightarrow^* w'''$. For the pseudocode see Algorithm~\ref{alg:cyk}.

\cykalgo

Time complexity of the algorithm is $\mathcal{O}(|P| + n \cdot |P| + n^3 \cdot |P| \cdot |N|^2)$ which equals to $\mathcal{O}(n^3)$ as we suppose the grammar fixed.

\textbf{CFL reachability}

Classic algorithm for (all-pairs) CFL reachability problem~\cite{10.1145/258994.259006, 10.5555/1196416} is a generalisation of the CYK algorithm. It is needed to find all $S$-reachable pairs of vertices in the given graph $D=(V, E, L), \, |V|=n$ for the given grammar $G=(N, \Sigma, P, S)$. The idea of the algorithm is based on the transitive closure, it iteratively inserts edges $v \xrightarrow{A} u$ in $D$ for all $A$-reachable pairs of vertices $(v, u)$ for every nonterminal $A \in N$. For the pseudocode see Algorithm~\ref{alg:cfl}.

\cflalgo

The algorithm correctness proof is based on the induction over the derivation length of the word on $S$-path between each pair of vertices. For establishing the time complexity of the algorithm it is enough to notice the following. Between every two vertices it can be no more than $|\Sigma| + |N|$ labeled edges, every triple in $W$ corresponds to the inserted edge and is added and deleted from $W$ at most once. The time complexity of the algorithm is $\mathcal{O}(|P| \cdot n + |P| \cdot n^2 + (|\Sigma| + |N|) \cdot n^2 \cdot n \cdot |P|) = \mathcal{O}(n^3)$ as in processing triple $(v, u, A)$ we iterate only over neighbours of $v$ and $u$.

\textbf{SAT to OV fine-grained reduction}~\cite{10.1016/j.tcs.2005.09.023}

Below we build a (SAT, $2^n$)$\rightarrow$(OV, $n^2$) fine-grained reduction. This is one of the basic and easy to understand reductions, moreover it implies that the OV conjecture is weaker than SETH, and thus harder to refute. This makes OV problem a good candidate to make a reduction from.

\underline{Reduction} Let $F$ be a SAT formula in $k$-conjunctive normal form with $n$ variables $x_1, \ldots, x_n$ and $m$ clauses $c_1, \ldots, c_m$. Without loss of generality $n \, \vdots \, 2$, otherwise we can add one more variable $x_{n+1}$ and a clause $c_{m+1} = (x_{n+1} \vee \ldots \vee x_{n+1})$, this operation does not affect the satisfiability of $F$. We build two sets of vectors $X, Y$ of the equal size $|X|=|Y|=2^{\frac{n}{2}}$, where each vector has dimension $m$. Each vector $x^{\alpha}$ from $X$ corresponds to the assignment $\alpha$ of values of $x_1, \ldots, x_{\frac{n}{2}}$. Then 

$$x^{\alpha}_i = 
\begin{cases}
1 &\text{if $x^{\alpha}$ satisfies clause $c_i$}\\
0 &\text{otherwise}
\end{cases}$$

Each vector $y^{\beta}$ from $Y$ corresponds to the assignment $\beta$ of values of $x_{\frac{n}{2} + 1}, \ldots, x_n$. Values of $y^{\beta}$ are defined in the same way as for $x^{\alpha}$:

$$y^{\beta}_i = 
\begin{cases}
1 &\text{if $y^{\beta}$ satisfies clause $c_i$}\\
0 &\text{otherwise}
\end{cases}$$

\underline{Correctness} Correctness of the reduction follows from the series of equivalent statements below. 

$F$ is satisfiable if and only if there exists an assignment $\gamma$ which satisfied every clause $c_i, i \in \{1, \ldots, m\}$. Let $\gamma = (\alpha', \beta')$, where $\alpha'$ assigns the values of $x_1, \ldots, x_{\frac{n}{2}}$, $\beta'$ assigns the values of $x_{\frac{n}{2} + 1}, \ldots, x_n$. For each clause $c_i$ it is true that $c_i$ is satisfied by $\gamma$ if and only if it is satisfied by at least one of $\alpha'$ or $\beta'$. The last observation is equivalent to the fact that $\forall i \in \{1, \ldots, m\} \; x^{\alpha'}_i \cdot y^{\beta'}_i = 0$, which exactly means orthogonality of vectors $x^{\alpha'}$ and $y^{\beta'}$.

The following theorem concludes the proof of correctness.

\begin{theorem}
SETH implies OV hypothesis.
\end{theorem}

\begin{proof}
Suppose OV problem can be solved in time $\mathcal{O}(n^{2 - \epsilon} \cdot poly(d))$ for some $\epsilon > 0$. 

Let $F$ be arbitrary $k$-CNF formula. From $F$ we create an OV instance as described above with $|X|=|Y|=2^{\frac{n}{2}}$ vectors in each set with dimension $m \le n^k$. Thus OV instance can be created in time $\mathcal{O}(2^{\frac{n}{2}} \cdot n^k)$. After that we solve OV problem with the existing algorithm in $\mathcal{O}((2^{\frac{n}{2}})^{2 - \epsilon} \cdot m^{\mathcal{O}(1)}) = \mathcal{O}(2^{n - \frac{\epsilon}{2}} \cdot n^{\mathcal{O}(k)}) = \mathcal{O}(2^{n - \epsilon'})$ time for arbitrary $k$. The answer to the OV instance is the answer to the original SAT instance. That implies that $k$-SAT problem can be solved in $\mathcal{O}(2^{\frac{n}{2}} \cdot n^k + 2^{n - \epsilon'}) = \mathcal{O}(2^{n - \epsilon''}), \epsilon''>0$ time for any $k$ which contradicts SETH. 
\end{proof}

\section{Overview of existing reductions}
\label{sec:map}
 
 \begin{figure}[!htp]
	
		\begin{center}  
			\includegraphics[scale = 0.9]{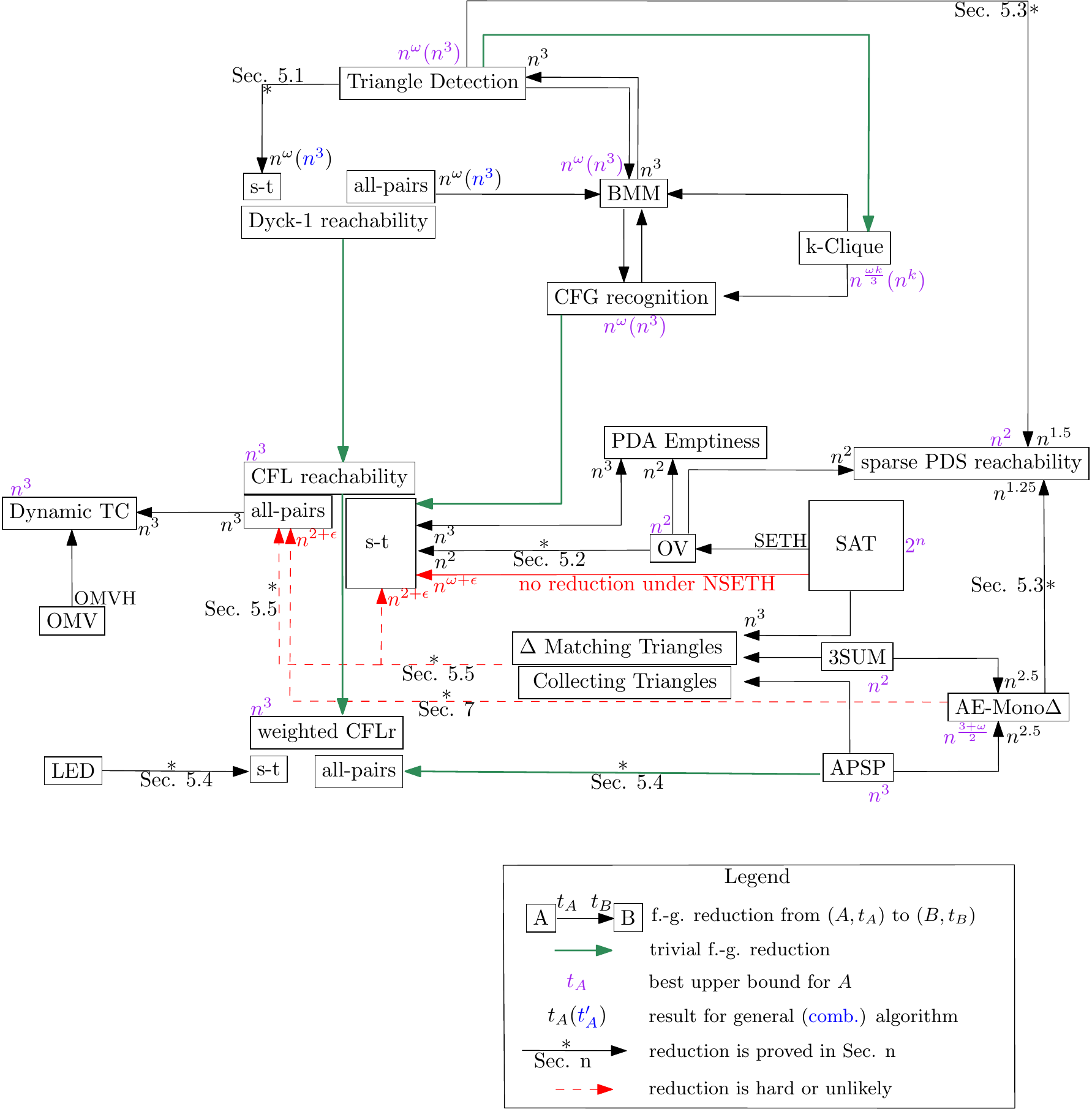}
		\end{center}
		
		\caption{Existing reductions concerning CFL reachability and CFG recognition. Arrow $a \rightarrow b$ represents a fine-grained reduction from $a$ to $b$. Note that asymptotics in upper bounds are written up to polylogarithmic (or polynomial in case of SAT) factors. When the asymptotic is not written on some side of the arrow it is equal to the upper bound of the problem (written in violet).}
		\label{fig:map}
		
	\end{figure}
	
	Before looking onto the reductions we need to mention the interconvertibility of CFL reachability problems and a class of set-constraint problems~\cite{10.1145/258994.259006} as this result allows us to reformulate our problem if we wish so.
	
	The CFL reachability problem has been shown to be complete for the class of two-way nondeterministic pushdown automata (2NPDA)~\cite{Rytter2005FastRO, 10.5555/788019.788876}. It means that subcubic algorithm for CFL reachability would lead to subcubic algorithms for the whole 2NPDA class and cubic upper bound has not been improved since the discovery of the class in 1968. Now we are ready to proceed to the fine-grained reductions.
	
	First of all we want to separate as subcase of CFL reachability --- Dyck-1 reachability problem. This subcase, when the grammar $G$ is fixed as the grammar of Dyck language on one type of parentheses, is known to be solvable in truly subcubic time and thus could be easier than the general CFL reachability problem.
	
	Concerning the upper bounds on this subcase one of the important reductions is a reduction from  all-pairs Dyck-1 reachability problem to BMM problem. It was firstly proved by Bradford~\cite{bradford2017efficient} via algebraic matrix encoding and then combinatorially by Mathiasen and Pavlogiannis~\cite{10.1145/3434315} by combining Dyck-1 path from bell-shaped paths. Via this reduction all-pairs (and s-t) Dyck-1 reachability can be solved in $n^3$ time by combinatorial algorithm and in $n^{\omega}$ time by general one. For this and the following reductions see Fig.~\ref{fig:map}.
	
	Through the subcubic combinatorial reductions between BMM and Triangle detection~\cite{Williams2009TriangleDV} and reduction from the latter to s-t Dyck-1 reachability (see Sec.~\ref{subsec:lower_from_tr} and~\cite{hansen2021tight}) we get that the above cubic upper bound is tight for combinatorial algorithms under BMM hypothesis. 
	That implies that in the world of combinatorial algorithms Dyck-1 case with high probability is not easier than the general one. Next we discuss the results for CFG recognition and general CFL reachability.
	
	$\mathcal{O}(n^{\omega})$ upper bound on CFG recognition was shown by Valiant in his famous paper~\cite{valiant1975general}. He reduced finding sets from CYK algorithm to several calls of BMM. The other way reduction created by Lee~\cite{10.1145/505241.505242} was combinatorial and thus allows us to get BMM based lower bounds on CFG recognition problem in both combinatorial and general case.  Combining the latter reduction with trivial reduction from CFG recognition to s-t CFL reachability we get another cubic combinatorial lower bound on our problem under BMM hypothesis, but this time we also achieve $\mathcal{O}(n^{\omega})$ conditional lower bound in general case.
	
    Another $\mathcal{O}(n^{\omega})$ lower bound on s-t CFL reachability was achieved via combination of reduction of Abboud et al.~\cite{abboud2018if} from $k$-Clique to CFG recognition and the reduction from CFG recognition to CFL reachability. This lower bound is based on $k$-Clique hypothesis. There are some reductions that are unlikely to help improve upper or lower bounds on the problem but they give insight on the essence of CFL reachability. 
	
	In the recent paper of Shemetova et al.~\cite{shemetova2021algorithm} the reduction from all-pairs CFL reachability to incremental DTC have been proven. Still this reduction cannot give truly subcubic algorithm for CFL reachability without refuting OMV conjecture~\cite{8948597, 10.1145/2746539.2746609}.
	
	Subcubic equivalence between s-t CFL reachability and PDA Emptiness problem was proven by Schepper~\cite{schepper2018complexity}. This equivalence allows us to reformulate the problem when we are proving cubic conditional lower bound on CFL reachability.
	
	Natural question that rises considering the lower bounds on CFL reachability is whether there exist one based on 3SUM, APSP or SAT hypotheses. Partially this question has been answered. Recently it was discovered by Chistikov et al.~\cite{10.1145/3498702} that there exist subcubic certificates for s-t CFL reachability (for existence and non-existence of the valid paths). From this fact it follows that there are no reductions under NSETH from SAT problem (SETH) to CFL reachability problem that give lower bound stronger than $\mathcal{O}(n^{\omega})$.
	
\end{section}

\section{Lower bounds}
\label{sec:lower}

This section is devoted to fine-grained reductions that produce lower bounds on CFL reachability problem or the problems connected to it.

\subsection{Lower bound based on BMM hypothesis}
\label{subsec:lower_from_tr}

\begin{figure}[!htp]
		
	\begin{center}  
		\includegraphics[scale=0.85]{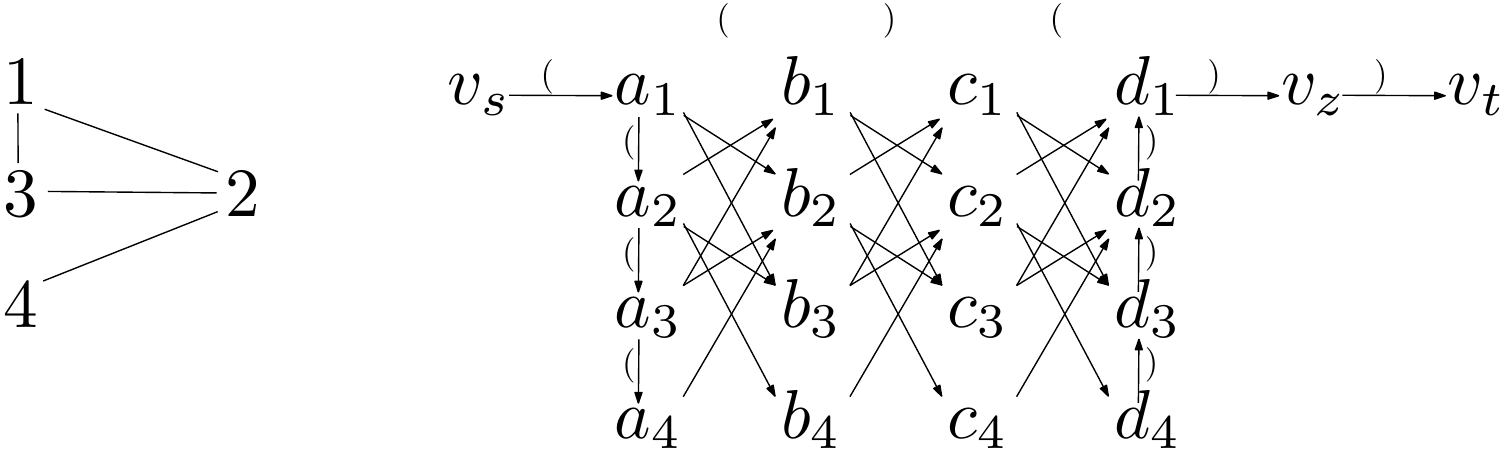}
	\end{center}
	
	\caption{The given graph in the triangle detection problem on the left and the corresponding graph with $(v_s, v_t)$-CFL Reachability problem on the right. Every vertex $i$ corresponds to vertices $a_i, b_i, c_i, d_i$; edges $(a_i, b_j), (b_i, c_j), (c_i, d_j)$ are marked with $(, ), ($ respectively.}
	\label{fig:triangle_detection_example}
	
\end{figure}

There exists a $\mathcal{O}(n^3)$ conditional combinatorial lower bound on s-t Dyck-1 reachability problem following from the fine-grained reduction from Triangle Detection problem. The reduction and the proof of its correctness is the same as the reduction from Triangle Detection to One-counter Net reachability problem (see Theorem 2 in ~\cite{hansen2021tight}). One-counter Net is a model in program verification isomorphic to Dyck-1 reachability. As formally the models are different, below we reformulate the result from~\cite{hansen2021tight} in terms of Dyck-1 reachability.

\underline{Reduction} Let $U = (V, E), V = \{x_1, \ldots, x_n\}$ be a graph in Triangle Detection problem. We build directed labeled graph $D = (V', E', L)$ for the Dyck-1 reachability problem as follows:

\begin{itemize}
    \item $V' = \{v_s\} \cup \{a_1, \ldots, a_n\} \cup \{b_1, \ldots, b_n\} \cup \{c_1, \ldots, c_n\} \cup \{d_1, \ldots, d_n\} \cup \{v_z\} \cup \{v_t\}$
    \item $(v_s, v_t)$ is a pair of vertices, for which we want to check whether it is Dyck-1-reachable
    \item $E' = \{(v_s, '(', a_1)\} \cup 
    \{(a_i, '(', a_{i+1})|i \in \{1, \ldots n - 1\}\} \cup$ 
    
    $\{(a_i, '(', b_j)|(x_i, x_j) \in E\} \cup
    \{(b_i, ')', c_j)|(x_i, x_j) \in E\} \cup$
    
    $\{(c_i, '(', d_j)|(x_i, x_j) \in E\} \cup
    \{(d_i, ')', d_{i-1})|i \in \{2, \ldots n\}\} \cup$
    
    $\{(d_1, ')', v_z)\} \cup \{(v_z, ')', v_t)\}$, where $(v, l, u)$ denotes a directed edge from $v$ to $u$ with label $l$ on it. Note that each undirected edge $(x_i, x_j)$ transforms into 6 directed edges, 2 between $a$- and $b$-vertices, $b$- and $c$-vertices, $c$- and $d$-vertices
    
    \item $L = \{'(', ')'\}$ that is exactly the Dyck-1 language alphabet 
\end{itemize}

For visual representation of the reduction see Fig.~\ref{fig:triangle_detection_example}.

\underline{Correctness} Let us first estimate the size of $D$:

\begin{itemize}
    \item $|V'| = 3 + 4 \cdot |V| = \mathcal{O}(n)$
    \item $|E'| = 3 + 2 \cdot (n - 1) + 6 \cdot |E| = \mathcal{O}(|E|) = \mathcal{O}(n^2)$
\end{itemize}

Thus the graph $D$ can be constructed in $\mathcal{O}(n^2)$ time.

\begin{lemma}
$v_t$ is Dyck-1-reachable from $v_s$ in $D$ if and only if $U$ contains a triangle.
\end{lemma}

\begin{proof}
We first prove the "if"-part of the claim. 

Suppose vertices $x_i, x_j, x_k$ form a triangle in $U$. Then in $D$ there exists a path 

$$v_s \rightarrow a_1 \rightarrow \ldots \rightarrow a_i \rightarrow b_j \rightarrow c_k \rightarrow d_i \rightarrow \ldots \rightarrow d_1 \rightarrow v_z \rightarrow v_t$$

The corresponding string of concatenated labels is 

$$( \overbrace{( \ldots (}^{\text{$i - 1$ times}} ()( \overbrace{) \ldots )}^{\text{$i - 1$ times}} )) = \overbrace{( \ldots (}^{\text{$i$ times}} ()() \overbrace{) \ldots )}^{\text{$i$ times}}$$

which is clearly a Dyck-1 word.

Now we proceed to the "only if"-part of the claim.

Suppose there exists a Dyck-1 path from $v_s$ to $v_t$. Note that $D$ is acyclic. Then from the construction of a graph we can see that the path should look like:

$$v_s \rightarrow a_1 \rightarrow \ldots \rightarrow a_l \rightarrow b_j \rightarrow c_k \rightarrow d_i \rightarrow \ldots \rightarrow d_1 \rightarrow v_z \rightarrow v_t$$

The corresponding word $w$, that we know is Dyck-1, is: 

$$( \overbrace{( \ldots (}^{\text{$l - 1$ times}} ()( \overbrace{) \ldots )}^{\text{$i - 1$ times}} )) = \overbrace{( \ldots (}^{\text{$l$ times}} ()() \overbrace{) \ldots )}^{\text{$i$ times}}$$ 

Then $w$ is Dyck-1 only if $l = i$. In that case by the construction of $E'$ we know, that in $U$ existed edges $(x_i, x_j), (x_j, x_k), (x_k, x_i)$. It implies that $x_i, x_j, x_k$ is a triangle in $U$.
\end{proof}

\begin{theorem}
\label{th:tr_to_dyck}
There exists a (Triangle Detection, $n^3$) to (Dyck-1 reachability, $n^3$) combinatorial fine-grained reduction. 
\end{theorem}

\begin{proof}
We can construct $D$ in time $\mathcal{O}(n^2)$. If Dyck-1 reachability problem is solvable in $\mathcal{O}(n^{3 - \epsilon})$ time, then Triangle Detection is solvable in $\mathcal{O}(n^{3 - \epsilon} + n^2) =\mathcal{O}(n^{3 - \epsilon'})$ time.
\end{proof}

For the general algorithms we can state a similar theorem. It is proved in an analogous way.

\begin{theorem}
There exists a (Triangle Detection, $n^{\omega}$) to (Dyck-1 reachability, $n^{\omega}$) fine-grained reduction.
\end{theorem}

Combining the Theorem~\ref{th:tr_to_dyck} with the combinatorial reduction from BMM to Triangle Detection~\cite{Williams2009TriangleDV} we get the following corollary.

\begin{cor}
If combinatorial BMM hypothesis is true, then there exist no combinatorial algorithm for Dyck-1 reachability problem with time $\mathcal{O}(n^{3 - \epsilon}), \forall \epsilon > 0$.
\end{cor}

\subsection{Lower bound based on OV conjecture}

This section is devoted to the (OV, $n^2$) to (s-t CFL reachability, $n^2$) fine-grained reduction. As the previous one this reduction was firstly presented in~\cite{hansen2021tight} but as the reduction from (OV, $n^2$) to (sparse PDS reachability, $n^2$) reduction. PDS reachability is a problem isomorphic to CFL reachability, still as formally the problems are different we reformulate the reduction in terms of s-t CFL reachability. In addition we change the reduction from~\cite{hansen2021tight} a bit so it translates to CFL reachability problem without $\epsilon$-edges.

We need to mention that this reduction gives nontrivial lower bound on s-t CFL reachability only in case of graphs with truly subquadratic number of edges, otherwise in graph with $n$ vertices and $\mathcal{O}(n^2)$ edges $\mathcal{O}(n^2)$ time is needed only to read the input graph.

We proceed with the construction of the reduction.

\underline{Reduction}
Suppose we are given two sets $X, Y$ of $n$ $d$-dimensional vectors as an instance of OV problem, $X = \{x_1, \ldots, x_n\}, \forall i \; x_i = (x^i_1, \ldots, x^i_d)$, $Y = \{y_1, \ldots, y_n\}, \forall i \; y_i = (y^i_1, \ldots, y^i_d)$. We construct an instance of weighted CFL reachability problem $D~=~(V, E, L, \Omega)$, $G = (N, \Sigma, P, S)$ as follows:

Let denote by $h_x$ the following function: 

$$h_x(b) = 
\begin{cases}
( &\text{if $b = 0$}\\
[ &\text{if $b = 1$}
\end{cases} \;\;\; \forall i \in \{1, \ldots, n\} \; \forall j \in \{1, \ldots, d\}$$

Let denote by $h_y$ the following function: 

$$h_y(b) = 
\begin{cases}
) &\text{if $b = 0$}\\
] &\text{if $b = 1$}
\end{cases} \;\;\; \forall i \in \{1, \ldots, n\} \; \forall j \in \{1, \ldots, d\}$$

\begin{itemize}
    \item $G$ is the grammar of Dyck-2 language
    \item $V = \{v_s, v_l, v_t\} \cup (\cup_{i = 1}^n \{u^i_j|j \in \{1, \ldots, d - 1\}\}) \cup (\cup_{i = 1}^n \{w^i_j|j \in \{1, \ldots, d - 1\}\})$
    \item $(v_s, v_t)$ is a pair of vertices for which we want to know if it is Dyck-2-reachable
    \item $E = E_{sl} \cup E_{lt}$, where
    
    $E_{sl} = (\cup_{i=1}^n (\{(v_s, h_x(x^i_1), u^i_1)\} \cup$
    
    $\{(u^i_j, h_x(x^i_{j + 1}), u^i_{j + 1})|j \in \{1, \ldots, d - 2\})\} \cup $
    
    $\{(u^i_{d - 1}, h_x(x^i_d), v_l)\}))$
    
    $E_{lt} = (\cup_{i=1}^n (\{(v_l, h_y(0), w^i_1)\} \cup \{(v_l, h_y(1), w^i_1|y^i_d = 0)\} \cup$
    
    $\{(w^i_j, h_y(0), w^i_{j + 1})|j \in \{1, \ldots, d - 2\}\} \cup$

    $\{(w^i_j, h_y(1), w^i_{j + 1})|y^i_{d-j} = 0 \; j \in \{1, \ldots, d - 2\}\} \cup$
    
    $\{(w^i_{d - 1}, h_y(0), v_t)\} \cup
    \{(w^i_{d - 1}, h_y(1), v_t)|y^i_1 = 0\}))$,
    
    where $(v, l, u)$ denotes a directed edge from vertex $v$ to vertex $u$ labeled with symbol~$l$.
\end{itemize}

For visual representation of the reduction see Fig.~\ref{fig:ov_to_cflr}.

\begin{figure}[!htp]
		
	\begin{center}  
		\includegraphics[scale=1.2]{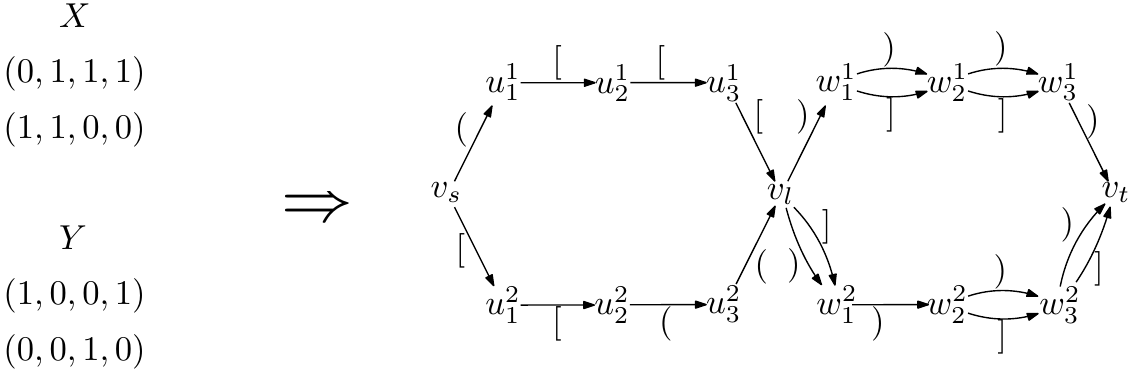}
	\end{center}
	
	\caption{The given sets $X, Y$ in the OV problem on the left and directed graph $D$ on the right.}
	\label{fig:ov_to_cflr}
	
\end{figure}

\underline{Correctness} Let us first estimate the size of $D$:

\begin{itemize}
    \item $V = 3 + 2 \cdot (d-1) \cdot n = \mathcal{O}(nd)$
    \item As each vertex of type $u^i_j, w^i_j, v_t$ has outgoing degree at most 2 and vertices $v_s, v_l$ have outgoing degree at most $2n$ we have $|E| = \mathcal{O}(nd)$.
\end{itemize}

Therefore the size of $D$ is $\mathcal{O}(nd)$ and it can be constructed in this time from OV instance.

\begin{lemma}
There exist two orthogonal vectors $x_i \in X, y_j \in Y$ if and only if $(v_s, v_t)$ is a Dyck-2-reachable pair of vertices in $D$.
\end{lemma}

\begin{proof}
We first prove the "only if"-part of the claim.

Suppose vectors $x_i \in X, y_j \in Y$ are orthogonal. We claim that there exists in $D$ a Dyck-2 path of the following construction (for now we do not choose one edge from multiple edges between consecutive vertices in the path, if there exist any):

\begin{equation}
v_s \rightarrow u^i_1 \rightarrow \ldots \rightarrow u^i_{d-1} \rightarrow v_l \rightarrow w^j_1 \rightarrow \ldots \rightarrow w^j_{d-1} \rightarrow v_t
\label{eq:path_constr}
\end{equation}

Note that the parenthesis on the outgoing edge from vertex $u^i_l$ should match with the parenthesis on the edge ingoing to vertex $w^j_{d-i}$. Case of vertices $v_s, v_t$ is similar and can be analyzed in the same way. Now we can build the Dyck-2 path iteratively: 

\begin{itemize}
    \item[-] If $x^i_l = 1$ we know that $y^j_l = 0$ as $x_i, y_j$ are orthogonal. By the construction an edge from $u^i_l$ is marked with "[" and one of the two edges to vertex $w^j_{d-i}$ is marked with "]". We choose these two edges with matching parentheses to our path.
    
    \item[-] If $x^i_l = 0$ we know that $y^j_l$ may be 0 and may be 1 as $x_i, y_j$ are orthogonal. By the construction an edge from $u^i_l$ is marked with "(" and in both cases ($y^j_l = 0$ and $y^j_l$ = 1) there exists an edge to vertex $w^j_{d-i}$ that is marked with ")". We choose these pair edges with matching parentheses to our path.
\end{itemize}

By the construction of the path it is a Dyck-2 path and $(v_s, v_t)$ is a Dyck-2-reachable pair of vertices in $D$.

Now we proceed to the "if"-part of the claim.

If $v_t$ is Dyck-2-reachable from $v_s$ in $D$ then there exists a Dyck-2 path between them. By the construction of $D$ the path is of the way described in~\eqref{eq:path_constr} for some $i$ and $j$. By the same observations as in the previous case we note that all pairs of parentheses on the edges should match. We prove for every pair of bits $x^i_m, y^i_m$ that they are orthogonal:

\begin{itemize}
    \item[-] If the corresponding pair of parentheses is a $()$-pair. Then $x^i_m = 0$ and $m$-th pair of bits of $x_i, y_j$ is orthogonal. 
    
    \item[-] If the corresponding pair of parentheses is a $[]$-pair, then $y^i_m = 0$, because only in this case we create an edge labeled with "]" ingoing to vertex $w^j_{d-m}$. Thus $m$-th pair of bits of $x_i, y_j$ is orthogonal.
\end{itemize}

We have proved that $x_i \in X, y_j \in Y$ are orthogonal.
\end{proof}

\begin{theorem}
There exists an (OV, $n^2$) to (s-t CFL reachability, $n^2$) fine-grained reduction.
\end{theorem}

\begin{proof}
Suppose we are given an instance of OV problem: two sets $X, Y$ of $n$ $d$-dimensional vectors. We construct an instance $D, G$ of s-t CFL reachability problem in time $\mathcal{O}(nd)$ as described above.

If there exists an $\mathcal{O}(n^{2-\epsilon})$-time algorithm for CFL reachability, then we can get an answer to an instance $D, G$ in time $\mathcal{O}((nd)^{2-\epsilon} = \mathcal{O}(n^{2-\epsilon} \cdot poly(d))$ and it is the answer to original OV instance. Thus we can solve OV problem in time $\mathcal{O}(n^{2-\epsilon} \cdot poly(d))$.
\end{proof}

As in the reduction we have build a sparse graph $D$ we can get the following corollary.

\begin{cor}
There in no $n^{2-\epsilon}, \epsilon >0$ algorithm for s-t CFL reachability on sparse graphs unless OV hypothesis is false.
\end{cor}

\subsection{Lower bound on push-down systems}
\label{subsec:pds}

Next we present a fine-grained combinatorial reduction (Triangle Detection, $n^3$) to (sparse PDS reachability with stack depth $b = \lceil \log n \rceil$, $n^{1.5}$). CFL reachability problem can be restated in terms of PDS reachability. In this section we want to add a restriction on stack depth which is a natural restriction in terms of PDS, therefore we use this problem.

\begin{figure}[!htp]
		
	\begin{center}  
		\includegraphics[scale=0.9]{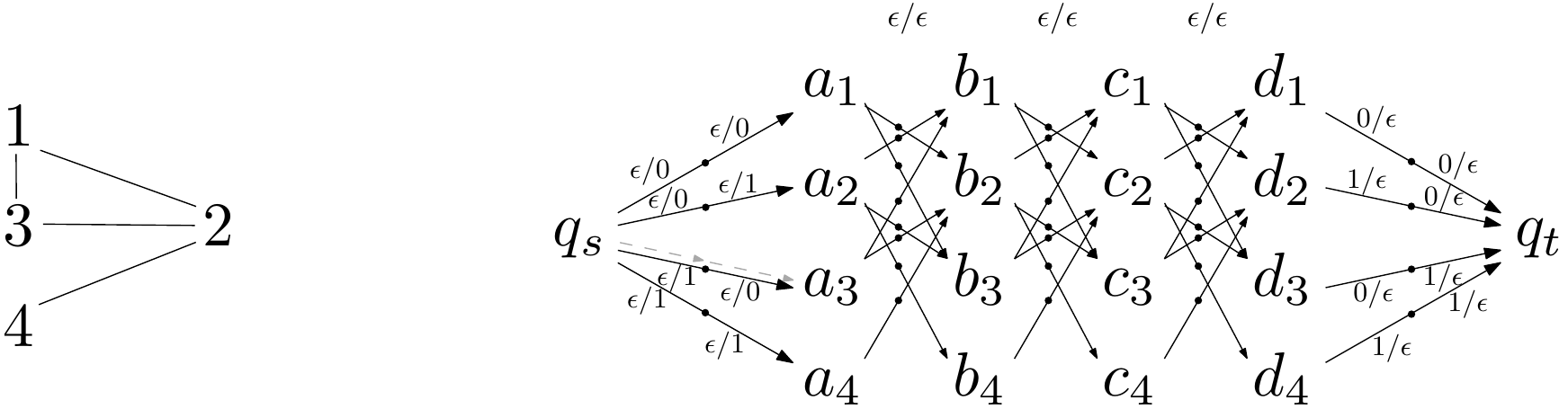}
	\end{center}
	
	\caption{The given graph in the triangle detection problem on the left and PDS $\mathcal{P}$ on the right. An edge from $v$ to $u$ labeled $x/y$ corresponds to a transition $(v,x,u,y) \in \sigma$. Each dot symbolizes auxiliary state, they split a transition into many transitions, e.g. see edge $(q_s, a_3)$.}
	\label{fig:triangle_to_pds}
	
\end{figure}

\underline{Reduction} 

Let $U = (V, E), V = \{x_1, \ldots, x_n\}$ be a graph in Triangle Detection problem. We build PDS $\mathcal{P} = (Q, \Gamma, \delta)$ as follows:

\begin{itemize}
    \item Set of states $Q$ consist of ordinary states $Q'$ and auxiliary states $Q''$. 
    
    \begin{itemize}
        \item[-] $Q' = \{q_s\} \cup \{a_1, \ldots, a_n\} \cup \{b_1, \ldots, b_n\} \cup \{c_1, \ldots, c_n\} \cup \{d_1, \ldots, d_n\} \cup \{q_t\}$.
        \item[-]$Q''$ are auxiliary states, they appear in transitions in which we want to read or write more than one symbol from the stack. If in transition $t$ we want to read $k_1$ symbols and write $k_2$ symbols, it is separated into $k_1 + k_2$ transitions. In each of the new transitions we either read one symbol from the stack without writing or write one symbol to the stack without reading. This procedure creates $k_1 + k_2 - 1$ auxiliary states. 
        
        Another auxiliary states can appear when we want to split a transition in two transitions. This type of auxiliary states will help to maintain the system sparse.
    \end{itemize}
     
    \item $q_s, q_t$ are the states, for which we want to check reachability 
    \item $\Gamma = \{0, 1\}$
    \item $\delta = 
    \{(q_s, \epsilon, [i]_2, a_i)|i \in \{1, \ldots n\}\} \cup$ 
    
    $\{(a_i, \epsilon, \epsilon, b_j)|(x_i, x_j) \in E\} \cup
    \{(b_i, \epsilon, \epsilon, c_j)|(x_i, x_j) \in E\} \cup$
    
    $\{(c_i, \epsilon, \epsilon, d_j)|(x_i, x_j) \in E\} \cup
    \{(d_i, \overline{[i]_2}, \epsilon, q_t)|i \in \{1, \ldots n\}\}$, where $[i]_2$ denotes a binary representation of $i$ (padded with leading 0-s if necessary) of length $\lceil \log n \rceil$, $\overline{[i]_2}$ denotes a reversed binary representation of $i$. Each transition which reads or writes $[i]_2$ is split in $\lceil \log n \rceil$ transitions with auxiliary states. Each transition which reads and writes $\epsilon$, nothing, is split in two transitions with one auxiliary state. 
\end{itemize}

For visual representation of the reduction see Fig.~\ref{fig:triangle_to_pds}.

\underline{Correctness} Let us first estimate the size of $\mathcal{P}$:

\begin{itemize}
    \item $|Q| = |Q'| + |Q''| \le (2 + 3 \cdot n) + (2 \cdot \log n + 3 \cdot n^2) = \mathcal{O}(n^2)$
    \item Each auxiliary state has exactly one transition from it. For each ordinary state transitions from it do not exceed $n$. As $|Q'| = \mathcal{O}(n)$, we get $|\delta| = \mathcal{O}(n^2) = \mathcal{O}(|Q|)$ 
\end{itemize}

Thus the system $\mathcal{P}$ can be constructed in $\mathcal{O}(n^2)$ time.

\begin{lemma}
$q_t$ is reachable from $q_s$ in $\mathcal{P}$ if and only if $U$ contains a triangle.
\end{lemma}

\begin{proof}
We first prove the "if"-part of the claim. 

Suppose vertices $x_i, x_j, x_k$ form a triangle in $U$. Then in $\mathcal{G}_{\mathcal{P}}$ there exists a path 

$$(q_s, \epsilon) \rightarrow \ldots \rightarrow (a_i, [i]_2) \rightarrow \ldots \rightarrow (b_j, [i]_2) \rightarrow \ldots$$

\begin{equation}
\ldots \rightarrow (c_k, [i]_2) \rightarrow \ldots \rightarrow (d_i, [i]_2) \rightarrow \ldots \rightarrow (q_t, \epsilon)
\label{eq:first}
\end{equation}

Now we proceed to the "only if"-part of the claim.

Graph $\mathcal{G}_{\mathcal{P}}$ is acyclic and by its construction path from $(q_s, \epsilon)$ to $(q_t, \epsilon)$ must be of the way described in equation \eqref{eq:first}. It implies  that $x_i, x_j, x_k$ form a triangle in $U$.
\end{proof}

\begin{theorem}
There exists a (Triangle Detection, $n^3$) to (sparse PDS reachability with stack depth $b = \lceil \log n \rceil$, $n^{1.5}$) combinatorial fine-grained reduction. 
\end{theorem}

\begin{proof}
By the observations above constructed PDS $\mathcal{P}$ is sparse and of size $|Q| = \mathcal{O}(n^2)$. stack depth is bounded by $[n]_2 = \lceil \log n \rceil$: we can only write on stack the number of one vertex in binary encoding.

$\mathcal{P}$ can be constructed in $\mathcal{O}(n^2)$ time. If there was $\mathcal{O}(n^{1.5 - \epsilon})$ algorithm for PDS reachability it would lead to $\mathcal{O}(n^2 + (n^2)^{1.5 - \epsilon}) = \mathcal{O}(n^{3 - \epsilon'})$ algorithm for Triangle Detection.
\end{proof}

We must mention that in~\cite{10.1145/3158118} it was proven conditional $\mathcal{O}(n^3)$ lower bound on PDS reachability based on BMM hypothesis, and by the similar technique of splitting edges it can be further translated into the conditional lower bound the was described above. Nevertheless the presented reduction is easier, has additional restriction on a stack depth and helps us in construction of the reduction below. The following conditional lower bound is based on 3SUM and APSP hypotheses.

We proceed with describing (AE-Mono$\Delta$, $n^{2.5}$) to (sparse all-pairs PDS reachability with stack depth $b = 4 \lceil \log n \rceil$, $n^{1.25}$) fine-grained reduction.

\underline{Reduction} Suppose we are given graph $U = (V, E), V = \{v_1, \ldots, v_n\}$ as an input in AE-Mono$\Delta$ problem. The construction of PDS is very similar to the one described above and based on the construction described in~\cite{10.1145/3571252} (see Theorem 3.2 in~\cite{10.1145/3571252}):

\begin{itemize}
    \item Set of states $Q$ analogously consists of ordinary states $Q'$ and auxiliary states $Q''$. As rules for the auxiliary states remain the same we describe only set of ordinary states. 
    
    $Q' = \{x_1, \ldots, x_n\} \cup \{y_1, \ldots, y_n\} \cup \{z_1, \ldots, z_n\} \cup \{x'_1, \ldots, x'_n\} \cup \{y'_1, \ldots, y'_n\}$.
     
    \item $(x_i, y'_j), i, j \in \{1, \ldots, n\}$ are the pairs of states, for which we want to check reachability, for that we will run all-pairs reachability algorithm
    \item $\Gamma = \{0, 1\}$
    \item $\delta = 
    \{(x_i, \epsilon, [i]_2 \, [j]_2 \, [c(e_{ij})]_2, y_j)|(x_i, x_j) \in E\} \cup$
    
    $\{(y_i, \overline{[c(e_{ij})]_2}, [c(e_{ij})]_2, z_j)|(x_i, x_j) \in E\} \cup$
    
    $\{(z_i, \overline{[c(e_{ij})]_2}, \epsilon, x'_j)|(x_i, x_j) \in E\} \cup$
    
    $\{(x'_i, \overline{[j]_2} \, \overline{[i]_2}, \epsilon, y'_j)|(x_i, x_j) \in E\}$, where $[i]_2$ denotes a binary representation of $i$ (padded with 0-s if necessary) of length $\lceil \log n \rceil$, $[c(e_{ij})]_2$ denotes a binary representation of number of color of the edge $e_ij = (v_i, v_j)$ in $U$ and is of length $2 \lceil \log n \rceil$  (padded with 0-s if necessary). 
\end{itemize}

For visual representation of the reduction see Fig.~\ref{fig:ae_mono_to_cflr}.

\underline{Correctness} Let us first estimate the size of $\mathcal{P}$:

\begin{itemize}
    \item $|Q| = |Q'| + |Q''| \le (5 \cdot n) + (4 \cdot (4 \log n + 1) \cdot n^2) = \mathcal{O}(n^2 \log n)$
    \item Each auxiliary state has exactly one transition from it. For each ordinary state transitions from it do not exceed $n$. As $|Q'| = \mathcal{O}(n)$, we get $|\delta| = \mathcal{O}(n^2) = \mathcal{O}(|Q|)$ 
\end{itemize}

Thus the system $\mathcal{P}$ can be constructed in $\mathcal{O}(n^2 \log n)$ time.

\begin{lemma}
For every pair of vertices $x_i, x_j$ there is a monochromatic triangle in $U$ containing an edge $e_{ij}$ if and only if $x_i$ is reachable from $y'_j$ in PDS $\mathcal{P}$.
\end{lemma}

\begin{proof}
We first proof the "only if"-part of the claim. 

Suppose there is a monochromatic triangle $x_i, x_j, x_k$ in $U$. Then there is a path in $\mathcal{G}_{\mathcal{P}}$:

$$(x_i, \epsilon) \rightarrow \ldots \rightarrow 
(y_j, [i]_2 \, [j]_2 \, [c(e_{ij})]_2) \xrightarrow{(1)} \ldots \rightarrow 
(z_k, [i]_2 \, [j]_2 \, [c(e_{jk})]_2) \xrightarrow{(2)} $$

$$\ldots \rightarrow 
(x'_i, [i]_2 \, [j]_2) \xrightarrow{(3)} \ldots \rightarrow 
(y'_j, \epsilon)$$
Note that as $c(e_{ij}) = c(e_{jk})$ transition (1) is valid, as $c(e_{jk}) = c(e_{ki})$ transition (2) is valid.

Now we proceed to the "if"-part of the claim.

Suppose there is a path from $(x_i, \epsilon)$ to $(y'_j, \epsilon)$ in $\mathcal{G}_{\mathcal{P}}$. By the construction of $\mathcal{P}$ and $\mathcal{G}_{\mathcal{P}}$ the path should be of the form:

$$(x_i, \epsilon) \rightarrow \ldots \rightarrow 
(y_j, [i]_2 \, [j]_2 \, [c(e_{ij})]_2) \xrightarrow{(1)} \ldots \rightarrow 
(z_k, [i]_2 \, [j]_2 \, [c(e_{jk})]_2) \xrightarrow{(2)} $$

$$(x'_{i'}, [i]_2 \, [j]_2) \xrightarrow{(3)} \ldots \rightarrow 
(y'_{j'}, \epsilon)$$

As transition (1) is valid we get that $c(e_{ij}) = c(e_{jk})$, as transition (2) is valid we get that $c(e_{jk}) = c(e_{ki})$, as transition (3) is valid we get that $i' = i, j' = j$. Thus $x_i, x_j, x_k$ is a monochromatic triangle in $U$.
\end{proof}

\begin{theorem}
There exists an (AE-Mono$\Delta$, $n^{2.5}$) to (sparse all-pairs PDS reachability with stack depth $b = 4 \lceil \log n \rceil$, $n^{1.25}$) fine-grained reduction. 
\end{theorem}

\begin{proof}
By the observations above constructed PDS $\mathcal{P}$ is sparse and of size $|Q| = \mathcal{O}(n^2 \log n)$. Stack depth is bounded by $[n]_2 = 4 \lceil \log n \rceil$: we write two numbers of vertices and a color number (which does not exceed $n^2$ as there at most $n^2$ edges in the graph) on the path from $x$-state to $y$-state, this gives $4 \lceil \log n \rceil$-upper bound on the depth of the stack. Transitions between $y$- and $z$-vertices first read $2  \lceil \log n \rceil$ symbols from the stack and then write the back, transitions between $z$- and $x'$-vertices and between $x'$- and $y'$-vertices only read symbols from the stack. This implies that our upper bound is correct.

$\mathcal{P}$ can be constructed in $\mathcal{O}(n^2 \log n)$ time. If there was $\mathcal{O}(n^{1.25 - \epsilon})$ algorithm for PDS reachability it would lead to $\mathcal{O}(n^{2} + (n^2)^{1.25 - \epsilon}) = \mathcal{O}(n^{2.5 - \epsilon'})$ algorithm for AE-Mono$\Delta$.
\end{proof}

As AE-Mono$\Delta$ has no $ \mathcal{O}(n^{2.5 - \epsilon}), \epsilon > 0$ algorithm unless both 3SUM and APSP hypotheses are false, we get the following corollary.

\begin{cor}
The is no algorithm for sparse all-pairs PDS reachability with stack depth $b = 4 \lceil \log n \rceil$ with time complexity  $\mathcal{O}(n^{1.25 - \epsilon}), \epsilon > 0$ if at least one of 3SUM and APSP hypotheses is true.
\end{cor}

\subsection{Reductions from LED variations}
\label{subsec:led}

Recall that LED problem is devoted to finding the minimum number of operations needed to transform the given word $w = w_1 \ldots w_n \in \Sigma^*$ into the word from the given CFL $\mathcal{L}(G), G = (N, \Sigma, P, S)$. It can be reformulated in a way, that every operation on a word has a cost: each symbol replacement costs $\omega_{repl}$, symbol deletion costs $\omega_{del}$, symbol insertion costs $\omega_{ins}$. The goal is to find transformation of $w$ into some $w' \in \mathcal{L}(G)$ of minimum total weight. In the canonical formulation of a problem $\omega_{repl} = \omega_{del} = \omega_{ins} = 1$.

LED problem is close to CFL reachability problem as it combines in itself CFG recognition with Edit Distance problem (for which we want to find distance from string $w$ to string $w'$). There are several cubic algorithms for LED~\cite{10.1137/0201022, 10.1016/0020-0190(95)00007-Y}. It is known of operations are restricted only to the insertions then a sub-cubic algorithm is unlikely to exists, this fact is proved via a reduction from APSP~\cite{DBLP:journals/corr/Saha14, 10.1145/3186893}. On the other hand, if insertions and deletions (in both cases: with or without replacement) are allowed, LED is solvable in truly subcubic time: there is $\tilde{\mathcal{O}}(n^{2.8603})$ algorithm for that~\cite{DBLP:journals/corr/BringmannGSW17}.

Next we present a reduction from LED to weighted s-t CFL reachability problem. The reduction works with every variation of LED: each operation may be allowed or restricted. Moreover different kind of operations may have different costs, if necessary. 

Let \textit{weighted CFL reachability problem}~\cite{lspidwnapew} be a generalisation of CFL reachability problem, where each edge in the given graph $D=(V, E, L, \Omega)$ has its integer weight $\omega$, let $\Omega$ be the set of weights of the edges. We suppose that there are no cycles in $D$ of total negative edge weight. Aside from the information about existence of $\mathcal{L}(G)$-path between some pairs of vertices in weighted CFL reachability problem we want to know the minimal possible weight of such path. There exists~\cite{lspidwnapew} an $\mathcal{O}(n^3)$ algorithm for weighted CFL reachability.

We allow edges to be labeled with $\epsilon$, the empty string. This assumption does not affect the complexity of a problem: we can change the grammar $G~=~(N, \Sigma, P, S)$ into the new grammar $G' = (N', \Sigma', P', S)$ in such a way that it operates with $\epsilon$ as with a new symbol of an alphabet. For that we suppose that $G$ is in CNF (otherwise we transform $G$ to it) and add a new nonterminal $E$ and a symbol $\epsilon'$ corresponding to $\epsilon$ ($\Sigma' = \Sigma \cup \{\epsilon'\}$), the empty string, and $2 \cdot |N| + 1$ new rules to productions $P$ to get production set $P'$:

\begin{itemize}
    \item[-] $E \rightarrow \epsilon'$
    \item[-] $A \rightarrow EA | AE \;\; \forall A \in N$
\end{itemize}

In this case new language $\mathcal{L}(G')$ contains exactly the words which, without symbol $\epsilon'$ in them, lie in $\mathcal{L}(G)$. In the given graph $D$ we change all edges with label $\epsilon$ to the same edges with label $\epsilon'$ and obtain a new graph $D'$. As $\epsilon'$ corresponds to the empty string the answer for CFL reachability problem for grammar $G'$ and graph $D'$ will be the same as the answer for CFL reachability problem for grammar $G$ and graph $D$.

Is it easy to see that CFL reachability problem is a subcase of its weighted version: we can assign each edge weight 1 (that way no negative cycles can appear) and look only on information about reachability in the answer for weighted CFL reachability.

It is also easy to see that APSP problem is a subcase of all-pairs weighted CFL reachability as we can define $G$ in such way that $\mathcal{L}(G) = \Sigma^*$, duplicate each edge and orient them in both ways and label the edges arbitrarily. In that case every path will be an $\mathcal{L}(G)$-path and we look for the path with minimal total edge weight between a pair of vertices. 

Now we proceed to (LED, $n^c$) to (weighted CFL reachability, $n^c$), $c > 1$ fine-grained reduction.

\underline{Reduction} Suppose we are given string $w = w_1\ldots w_n \in \Sigma^*$ and CFG $G = (N, \Sigma, P, S)$ as an instance of LED problem. We construct an instance of weighted CFL reachability problem $D = (V, E, L, \Omega)$, $G'$ as follows:

\begin{itemize}
    \item $G' = G$
    \item $V = \{s = v_1, \ldots, v_{n+1} = t\}$
    \item $s = v_1, v_{n+1} = t$ are vertices for which we want to know the minimal weight of $\mathcal{S}$-path between them
    \item $E = \{(v_i, s(\omega_{ins}), v_i)|s \in \Sigma, i \in \{1, \ldots, n+1\}\} \cup$ 
    
    $\{(v_i, \epsilon(\omega_{del}), v_{i + 1})|i \in \{1, \ldots, n\}\} \cup$
    
    $\{(v_i, s'(\omega_{repl}), v_{i + 1})|s' \in \Sigma \setminus \{w_i\}, i \in \{1, \ldots, n\}\}$, where $(v, s(\omega), u)$ denotes a directed edge from vertex $v$ to vertex $u$ labeled with $s$ and having weight $\omega$. Edges corresponding to the insertion procedure have weight $\omega_{ins}$ (\emph{insertion edges}), \emph{replacement edges} (have weight $\omega_{repl}$) and \emph{deletion edges} (have weight $\omega_{del}$) are defined analogously.
\end{itemize}

For visual representation of the reduction see Fig.~\ref{fig:led}.

\begin{figure}[!htp]
		
	\begin{center}  
		\includegraphics[scale=1]{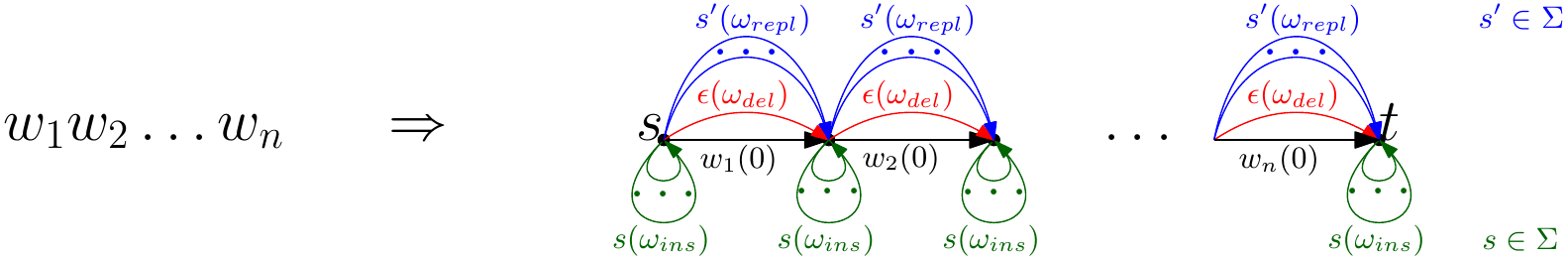}
	\end{center}
	
	\caption{Reduction from LED to CFL reachability with weights. On every edge mark $x(\omega)$ means that the edge is labeled with symbol $x \in \Sigma$ and is assigned weight $\omega$.}
	\label{fig:led}
	
\end{figure}

\underline{Correctness} The size of the created directed graph is linear:

\begin{itemize}
    \item $|V| = n + 1$
    \item $|E| = n \cdot (|\Sigma| + 1) + (n + 1) \cdot |\Sigma| = \mathcal{O}(n)$
\end{itemize}

Thus $D$ can be constructed in linear time.

Note that is some operations are forbidden, we can exclude the corresponding (insertion, deletion, replacement) edges from the graph, the upper bound on size of $D$ will remain the same.

\begin{lemma}
\label{th:led}
Let $k$ be the answer to instance $w, G$ of LED problem, $k'$ be the answer to constructed from $w$ above instance $D, G'$ of s-t weighted CFL reachability problem, then $k = k'$.
\end{lemma}

\begin{proof}
We first prove that $k' \le k$. 

Let $w'$ be the closest string to $w$ and $T$ be the list of transformations needed to transform $w$ to $w'$. We take as an initial $S$-path from $s$ to $t$ a zero-weight path $P$ over edges labeled with symbols of $w$. After that we apply every transformation from $T$ to our path:

\begin{itemize}
    \item[-] If transformation inserts a symbol $s$ after $w_j$ we add an edge $(v_{j+1}, s(\omega_{ins}), v_{j+1})$ to our path. If a symbol is inserted in the beginning of the string we add an edge $(v_{0}, s(\omega_{ins}), v_{0})$.
    \item[-] If transformation deletes a symbol $w_j$ we replace an edge $(v_{j - 1}, w_{j}(0), v_{j})$ with an edge $(v_{j-1}, \epsilon(\omega_{del}), v_{j})$ in our path.
    \item[-] If transformation replaces a symbol $w_j$ with $s'$ we replace an edge $(v_{j - 1}, w_{j}(0), v_{j})$ with an edge $(v_{j-1}, s'(\omega_{repl}), v_{j})$ in our path.
\end{itemize}

If several symbols are inserted between $w_i$ and $w_{i+1}$ we suppose that the loops that were added in vertex $v_{i+1}$ in path $P$ are traversed by $P$ in the same order as they are in $w'$.

By construction the path $P$ now is an $S$-path between $s$ and $t$, which labels form the word $w'$. $P$ has weight $k$. That means that $k' \le k$.

We proceed to the second part of the proof: $k \le k'$.

Suppose there is an optimal $S$-path of weight $k'$ --- $P$ between $s$ and $t$, which labels form the word $w' \in \mathcal{L}(G)$. We translate it to the list of transformations of $w$ to obtain $w'$.

\begin{itemize}
    \item[-] Edge $(v_{i}, s(\omega_{ins}), v_{i})$ translates into insertion of $s$ after $w_{i - 1}$ or before $w_1$ if $i = 0$.
    \item[-] Edge $(v_{j-1}, \epsilon(\omega_{del}), v_{j})$ in $P$ translates into deletion of symbol $w_j$ in $w$.
    \item[-] Edge $(v_{j-1}, s'(\omega_{repl}), v_{j})$ in $P$ translates into replacement of $w_j$ with $s'$ in $w$.
\end{itemize}

If there were several loops through $v_i$ in path $P$ they are processed in the inversed order from the order in path. With this construction after application of all transformations to $w$ we will get $w'$ and the cost of the transformation is exactly $k'$. It implies that $k \le k'$.
\end{proof}

Once again note, that if some operation are restricted Lemma~\ref{th:led} works.

We get the following results.

\begin{theorem}
There is a (LED, $n^c$) to (weighted s-t CFL reachability, $n^c$), $c > 1$ fine-grained reduction.
\end{theorem}

\begin{proof}
The construction of $D$ takes $\mathcal{O}(n)$ time. If weighted s-t CFL reachability is solvable in $\mathcal{O}(n^{c - \epsilon}), 0 < \epsilon < c - 1$ time, then LED is solvable in $\mathcal{O}(n^{c - \epsilon'}), 0 < \epsilon' < c - 1$ time.
\end{proof}

In case when insertions are the only allowed operations, combining the conditional lower bound on LED based on APSP hypothesis~\cite{DBLP:journals/corr/Saha14, 10.1145/3186893} with the existence of subcubic reduction from LED to weighted s-t CFL reachability we get the following corollary.

\begin{cor}
There is no $\mathcal{O}(n^{3 - \epsilon}), \epsilon > 0$ algorithm for weighted s-t CFL reachability problem under APSP hypothesis.
\end{cor}

It is worth mentioning that this result applies to the subcase when the weights in weighted CFL reachability problem are restricted to only 0 and 1, i.e. there is no truly subcubic algorithm for small weights (integer weights with absolute value bounded by some constant $M$)  CFL reachability problem. For the analogous small weight version of APSP problem (in both directed and undirected case) there are known~\cite{10.1006/jcss.1997.1388} truly subcubic algorithms with time complexity $\mathcal{O}(n^{\frac{3+\omega}{2}})$. Together these two results imply that small weights CFL reachability problem in terms of fine-grained complexity is strictly harder than small weights APSP problem if APSP hypothesis it true.

\subsection{Other possible reductions}
\label{subsec:other_red}

The SAT, APSP and 3SUM problems are ones the most popular to get a reduction from. We have considered three problems which lower bounds are based on the SAT, APSP and 3SUM hypotheses as a possible candidates for creating a reduction.

Collecting Triangles and $\Delta$ Matching Triangles are a pair of problems, both of them have no subcubic lower bound unless all three of SETH, 3SUM and APSP hypotheses are false~\cite{10.1145/2746539.2746594}.

These problems are good in a way that a reduction from one of them would create the most believable lower bound for CFL reachability. Nevertheless creating a reduction from them seems unlikely.

\begin{theorem}
There is no reduction from (Collecting Triangles/$\Delta$ Matching Triangles, $n^3$) to (CFL reachability, $n^c, c > 2$) unless at least one of the following is false:

\begin{itemize}
    \item $\omega = 2$, where $\omega$ is the matrix multiplication exponent
    \item NSETH
    \item There exists an (all-pairs CFL reachability, $n^{c}$) to (s-t CFL reachability, $n^{c'}, c' > 2$) fine-grained reduction.
\end{itemize}
\end{theorem}

\begin{proof}
First of all note that if there exists a reduction to the s-t CFL reachability, then we get an $n^c, c>2$ lower bound via a reduction from SAT. It is known~\cite{10.1145/3498702} that no lower bound on s-t CFL reachability better than $n^{\omega}$ can be achieved via a reduction from SAT unless NSETH is false. This implies that either $\omega > 2$ or NSETH is false.

Let us proceed to the case when a reduction was to the all-pairs CFL reachability. If there exists (all-pairs CFL reachability, $n^{c}$) to (s-t CFL reachability, $n^{c'}, c' > 2$) fine-grained reduction then we are in the previous case. Combining that reduction with the given reduction from Collecting Triangles/$\Delta$ Matching Triangles and with the reduction to the latters from SAT we get the desired reduction from SAT to s-t CFL reachability.
\end{proof}

Another problem that we have looked onto was AE-Mono$\Delta$. As its $\mathcal{O}(n^{2.5})$ conditional lower bound~\cite{https://doi.org/10.4230/lipics.itcs.2020.53} is based on 3SUM and APSP hypotheses, reduction from it cannot refute NSETH. We show in the end of Sec.~\ref{sec:line_edges} that the reduction from this problem to all-pairs CFL reachability is either nontrivial or highly unlikely.

\section{CFL reachability with bounded paths length}
\label{sec:bounded_paths}

Is is known that s-t CFL reachability on directed acyclic graph (DAG) can be solved faster that in general case:

\begin{theorem}[\cite{schepper2018complexity}, Th.5.9]
The s-t CFL reachability on a DAG $H = (V, E)$ can be solved in $\mathcal{O}(|H|^{\omega})$ time for a fixed grammar.
\end{theorem}

The proof of this theorem is based on generalisation of Valiant's parser for CFG recognition problem. Valiant parser works as follows. It builds an $n \times n$ upper triangular matrix $M$ where the cell $(i, j)$ contains the set $T_{i, j}$ from CYK algorithm. Matrix is build through computation of the transitive closure of the matrix of CYK sets in the beginning of the algorithm. 

It was noticed in~\cite{schepper2018complexity} that the vertices in DAG have topological sorting and if we number the vertices according to it, the analogous $M$ matrix is upper triangular. Thus Valiant's algorithm applies to DAGs too with minor corrections. Recall that topological sorting can be built in $\mathcal{O}(|V| + |E|) = \mathcal{O}(n^2)$ and sorting vertices according to it does not affect the time complexity of the algorithm.

With some further observations we get that after computing transitive closure $M$ contains information about reachability of all pairs of vertices and we can formulate the following theorem.

\begin{theorem}
\label{th:cflr_on_dag}
The all-pairs CFL reachability problem on a DAG $H$ can be solved in $\mathcal{O}(|H|^{\omega})$ time for a fixed grammar.
\end{theorem}

Next we use Theorem~\ref{th:cflr_on_dag} to obtain faster all-pairs CFL reachability algorithm for finding paths of bounded length. 

\begin{theorem}\label{lemma:constant-paths}
There exists an algorithm for all-edges CFL reachability with time complexity $\tilde{\mathcal{O}}(n^{\omega})$ that finds all $S$-reachable pairs of vertices that are reachable with paths of length no more than $k = \tilde{\mathcal{O}}(1)$.
\end{theorem}

\begin{figure}[!htp]
		
	\begin{center}  
		\includegraphics[scale=1]{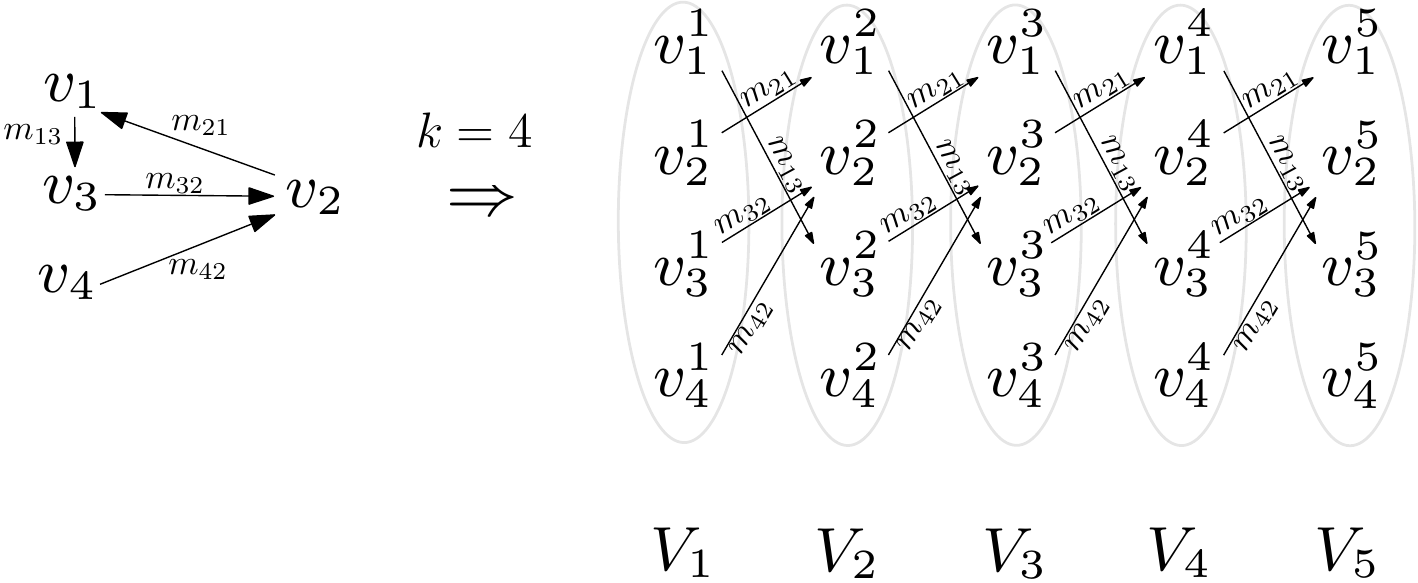}
	\end{center}
	
	\caption{Example of the reduction on graph with $4$ vertices and searching for CFL paths of length no more than $4$.}
	\label{fig:k-paths}
	
\end{figure}

\begin{proof}
Let $D = (V, E, L), V = \{v_1, \ldots, v_n\}$ be the given graph in all-pairs CFL reachability problem. We build a new labeled DAG $H = (V', E', L')$ of size $k = \tilde{\mathcal{O}}(n)$ as follows.

\begin{itemize}
    \item $V' = V_1 \cup \ldots \cup V_k$, where $V_i = \{v^i_1, \ldots, v^i_n\}, i \in \{1, \ldots, k\}$ is a copy of set of vertices $V$.
    \item $E' = \cup_{i = 1}^{k - 1} \{(v^i_p, m_{pl}, v^{i+1}_l)|(v_p, m_{pl}, v_l) \in E)\}$, where $(v_p, m_{pl}, v_l)$ denotes directed edge from $v_p$ to $v_l$ with label $m_{pl}$.
\end{itemize}

For visual representation of the reduction see Fig.~\ref{fig:k-paths}.

The number of vertices in graph $H$ is $|V'| = k \cdot n = \tilde{\mathcal{O}}(n)$. Notice that $H$ is acyclic. Now we can apply Theorem~\ref{th:cflr_on_dag} to $H$.  

We extract the reachability information for each pair of vertices $v_p, v_l$ as follows. $v_l$ is $S$-reachable from $v_p$ in $G$ via a path of length at most $k$ if at least one pair of vertices $(v_p^1, v_l^i), i \in \{1, \cdots, k\}$ is $S$-reachable in $H$. We need $\mathcal{O}(k)$ time to extract bounded-reachability information for every pair of vertices in $G$ from reachability in $H$. Thus we can get all-pairs reachability information in $\tilde{\mathcal{O}}(n^2)$ and it does not affect time complexity of our algorithm.

Time complexity of the algorithm is equal to the time complexity of application Theorem~\ref{th:cflr_on_dag} to $H$, which is $\mathcal{O}((n \cdot poly(\log n))^{\omega})$ = $\tilde{\mathcal{O}}(n^{\omega})$.
\end{proof}

We can extend the theorem above to the paths of greater lengths to get an algorithm with worse but still subcubic time complexity.

\begin{cor}\label{cor:paths}
There exists an algorithm for all-pairs CFL reachability with time complexity $\tilde{\mathcal{O}}(n^c), c < 3$ that finds all $S$-reachable pairs that are reachable with paths of length no more than $k = \tilde{\mathcal{O}}(n^{c'})$, where $c' < \frac{3}{\omega} - 1$.
\end{cor}

\begin{proof}
We build graph $H$ in the same way as in Theorem~\ref{lemma:constant-paths}. Graph $H$ has $k \cdot n$ vertices.

Time complexity of application Theorem~\ref{th:cflr_on_dag} to $H$ is: $\mathcal{O}((k \cdot n)^{\omega}) = \tilde{\mathcal{O}}((n^{1 + c'})^{\omega}) = \tilde{\mathcal{O}}((n^{(1 + c') \cdot \omega})$. As $(1 + c') \cdot \omega < 3$ it is subcubic.

Time complexity of extraction the reachability information from $H$ is: $\mathcal{O}(n^2 \cdot (k - 1)) = \tilde{\mathcal{O}}(n^{2 + c'})$. As $2 + c' < (1 + c') \cdot \omega < 3$ it is subcubic.

Total time complexity of the algorithm is $\tilde{\mathcal{O}}(n^c)$ for $c = (1 + c') \cdot \omega < 3$.
\end{proof}

Note that in CFL reachability problem some paths that we look for may be of exponential length~\cite{PIERRE1992279}. It implies that Theorem~\ref{lemma:constant-paths} and Corollary~\ref{cor:paths} can give us a faster algorithm in a very restricted case.

\section{CFL reachability on graphs with line-edges}
\label{sec:line_edges}

Suppose we want to solve all-pairs CFL reachability problem on a graph $D$ that is homeomorphic to graph $D', |V(D')|=n$ (i.e. $D$ is a subdivision of $D'$) but our vertices of interest in all-pairs reachability are exactly the vertices of corresponding to the vertices of $D'$. An example of such was in Section~\ref{subsec:pds}, see Fig.~\ref{fig:triangle_to_pds}: we have split the transition edges to make the system sparse, but the reachability information between auxiliary states was irrelevant to us. 

In general graphs even if the number of subdivisions that we can make on each edge is bounded by constant, number of vertices can increase from $\mathcal{O}(n)$ up to $\mathcal{O}(n^2)$. In that case the application of the standard algorithms for CFL reachability will take $\mathcal{\tilde O}(n^6)$ amount of time. We show how to preprocess $D'$ and the grammar if the number of subdivisions on each edge was small enough so that the time complexity of the algorithm will remain the same with respect to the polylogarithmic factors.

We say that $D$ is a \textit{subdivision-graph} if $D$ is a subdivision of graph $D'$. If every edge of $D'$ was subdivided less than $k$ times, $D$ is called \textit{$k$-subdivision-graph}

In $D$ we call vertices of $D'$ \textit{ordinary} vertices and vertices that were created through subdivision \textit{additional} vertices. Edge $e$ of $D'$ after subdivision is called \textit{line-edge} of $D$.

\emph{All-pairs CFL reachability on subdivision-graph $D$} is a variation of all-pairs CFL reachability problem on $D$ where we ask only about reachability between pairs of ordinary vertices.

$k$-subdivision-graph can rise up if the mark that we want to put on each edge does not consist of one symbol, but rather a word (see Fig.~\ref{fig:triangle_to_pds} and edges from $q_s$ for an example). In that case, a logical step would be to replace an edge with a word on it with a subdivision-graph, which symbols form this word. In this case $k$ equals the biggest length of the word on the edges of $G$.

\begin{theorem}
\label{th:line-edges-1}
Suppose there exists an all-pairs CFL reachability algorithm with time complexity $\mathcal{O}(n^c \cdot poly(|G|)) = \mathcal{O}(n^c)$. Then there exist all-pairs CFL reachability algorithm on $k$-subdivision-graphs, $k = o(\log n)$ with time complexity $\mathcal{O}(n^{c+c'}), \forall c' > 0$.
\end{theorem}

\begin{proof}
Let $D = (V, E, L)$ be an $k$-subdivision-graph graph, $k = o(\log n)$, on $n$ ordinary vertices and $G = (N, \Sigma, P, S)$ be the grammar in the given instance of all-pairs CFL reachability problem. We transform this instance to new instance $D', G'$ of the same problem. We suppose that $G$ is in CNF, otherwise, it can be converted to CNF in $\mathcal{O}(poly(|G|))$ time.

The proof consists of four steps that transform the words on the line-edges so that they are the terminals in the new grammar $G'$. $D'$ in the new instance is a graph, which was transformed to $D$ by subdivisions of the edges, i.e. $D'$ is a graph on ordinary vertices of $D$, vertices in $D'$ are connected by an edge if the corresponding vertices in $D$ are connected by an edge or a line-edge. Labels on the edges of $D'$ are defined below together with the grammar $G'$.

The steps of the transformation from an instance $D, G$ to an instance $D', G'$ are the following:

\begin{enumerate}
    \item First of all we make all line-edges in $D$ of the same length. 
    
    As $D$ is a $k$-subdivision-graph all line-edges of $D$ consist of no more than $k$ usual edges. We want to insert $\epsilon$-edges in every line-edge so that the length of every line-edge is exactly $k = o(\log n)$, the resulting graph is denoted by $D''$. We transform the grammar $G$ to $G''$ accordingly, see Sec.~\ref{subsec:led}, note about adding $\epsilon$-edges in the graph for the full description.
    
    Note that all words from $\mathcal{L}(G'')$ that we look for in all-pairs CFL reachability problem in $D$ now have length $k \cdot i$ for some $i \ge 0$.
    
    \item We define the new alphabet $\Sigma'$ as an alphabet of $k$-tuples of symbols over $\Sigma \cup \{\epsilon'\}$: $\Sigma' = \{(x_1, \ldots, x_k)|x_i \in \Sigma \cup \{\epsilon'\} \; \forall i\in \{1, \ldots, k\}\}$. We mark an edge in $D'$ by the $k$-tuple $(x_1, \ldots, x_k)$ from $\Sigma'$ if the marks on the corresponding transformed line-edge in $D''$ had marks $x_1, \ldots, x_k$ on the edges.

    \item We convert CFG $G''$ to PDA $\mathcal{A}$ with $l = poly(|G''|)$ number of states by a standard algorithm (see~\cite{10.5555/524279}, Lemma 2.21)

    \item After that we transform PDA $\mathcal{A}$ to PDA $\mathcal{A}'$ with $l \cdot |\Sigma|^{k} \cdot k = \mathcal{O}(n^{c'}), \forall c' > 0$ states that recognises all words from $\mathcal{L}(G'')$ that have length $k \cdot i$ for some $i$ and were split into $k$-tuples. That means that $\mathcal{A}'$ accepts the words over $k$-tuple language of the form $(w_1, \ldots, w_k)\ldots(w_{ik+1}, \ldots, w_{(i+1)k})$ if and only if $\mathcal{A}$ accepted the word $w_1\ldots w_{(i+1)k}$.
    
    We transform each state $q$ except $q_0$ of $\mathcal{A}$ into $|\Sigma|^{k} \cdot k$ states. We suppose that $\mathcal{A}$ can be in state $q_0$ only in the beginning of its work, otherwise we can make a duplicate state $q_0'$ with the same transitions as $q_0$ where $\mathcal{A}$ will come in the middle of its work. State $q_{(w_1, \ldots, w_k)}^i$ corresponds to the moment of the system when the automaton $\mathcal{A}$ has read the $i$-th symbol of the tuple (if they were written as a part of usual string $\ldots w_1 \ldots w_k \ldots$) and is in state $q$. 
    
    For each transition $(q, s, g, q', g') \in \delta$ we model the work of $\mathcal{A}$ and create the following transitions:
    
    $$(q_{(w_1, \ldots, w_k)}^i, \epsilon, g, q'^{i+1}_{(w_1, \ldots, w_k)}, g') \;\; \forall 0 \le i < k \;\; \forall (w_1, \ldots, w_{i+1} = s, \ldots, w_k) \in \Sigma'$$
    
    $$(q_{(w_1, \ldots, w_k)}^k, (w'_1, \ldots, w'_k), g, q'^{0}_{(w'_1, \ldots, w'_k)}, g') \;\; \forall (w_1, \ldots, w_k), (w'_1 = s, \ldots, w'_k) \in \Sigma'$$
    
    The set of final states consists of states: 
    
    $$q^k_{(w_1, \ldots, w_k)} \;\; \forall (w_1, \ldots, w_k) \in \Sigma' \;\; \forall q \in Q_f$$
    
    where $Q_f$ is a set of final states of $\mathcal{A}$.
    
    Transitions in the cases concerning the starting state $q_0$ are similar. For transition $(q_0, s, g, q', g') \in \delta$ in $\mathcal{A}$ we create in $\mathcal{A'}$ transitions:
    
    $$(q_0, (w_1, \ldots, w_k), g, q'^{0}_{(w_1, \ldots, w_k)}, g') \;\; \forall (w_1 = s, \ldots, w_k) \in \Sigma'$$
    
    By the construction $\mathcal{A'}$ makes the same transitions as $\mathcal{A}$, but "saves" a tuple that is read now in a memory.

    \item Finally we convert PDA $\mathcal{A}'$ to CFG $G'$ of size  $|G'| = \mathcal{O}(n^{c'}), \forall c' > 0$ by a standard algorithm (see~\cite{10.5555/524279}, Lemma 2.22)
\end{enumerate}

Now one can apply algorithm for CFL reachability with grammar $G'$ for graph $D'$ and get the answer to the usual all-pairs CFL reachability problem in time $\mathcal{O}(n^c \cdot poly(|G'|)) = \mathcal{O}(n^c \cdot poly(n^{c'})) = \mathcal{O}(n^{c+c'}), \forall c' > 0$. Note that $\mathcal{L}(G')$-reachable pairs of $D'$ correspond exactly to the $\mathcal{L}(G)$-reachable pairs of $D$.  

\end{proof}

Note that Theorem~\ref{th:line-edges-1} can be seen as a reformulation of a fact that if a grammar has size $o(n)$ then it does not affect the time complexity of an algorithm if it's time complexity depends polynomially on the size of the grammar, which is typical for CFL reachability algorithms. The result is useful when the words on the edges cannot be easily expressed as nonterminals of some grammar. 

One can ask a question whether Theorem~\ref{th:line-edges-1} can be generalized for the line-edges of the bigger length. Below we show that this kind of generalisations will be either hard to prove or highly unlikely for bigger lengths.

Let $A$ be the algorithm for all-pairs CFL reachability with time complexity $\mathcal{O}(T(n)~\cdot~poly(|G|))$. As generalisation should be correct for any time complexity $T(n)$ of algorithm $A$ we suppose that in the proof of generalisation of the theorem above, one makes several calls to algorithm $A$ on the graphs of size $\Theta(n)$. We suppose that these calls are made on directed graphs $D_1, \ldots, D_k, k \ge 1$. By the statement of the theorem we get that $k = \mathcal{O}(poly(\log n))$. With this assumption we can formulate the following theorem.

\begin{theorem}\label{lemma:no_expansion}
Suppose claim of Theorem~\ref{th:line-edges-1} holds for the $k$-subdivision-graphs, $k = \Omega(\log n)$. Then either:

\begin{itemize}
    \item at least one of $D_1, \ldots, D_k, k \ge 1$ has structure different from $D'$
    \item both APSP and 3SUM hypotheses are false
\end{itemize}
\end{theorem}

Before the proof of the theorem we present a reduction which was proposed in~\cite{10.1145/3571252}.

\begin{lemma}{\cite{10.1145/3571252}}\label{lemma:no_exp_help}
There exists an (AE-Mono$\Delta$, $n^3$) to (all-pairs CFL reachability on $k$-subdivision-graph, $k = \Theta(\log n)$, $n^3$) fine-grained reduction.
\end{lemma}

\begin{figure}[!htp]
		
	\begin{center}  
		\includegraphics[scale=1]{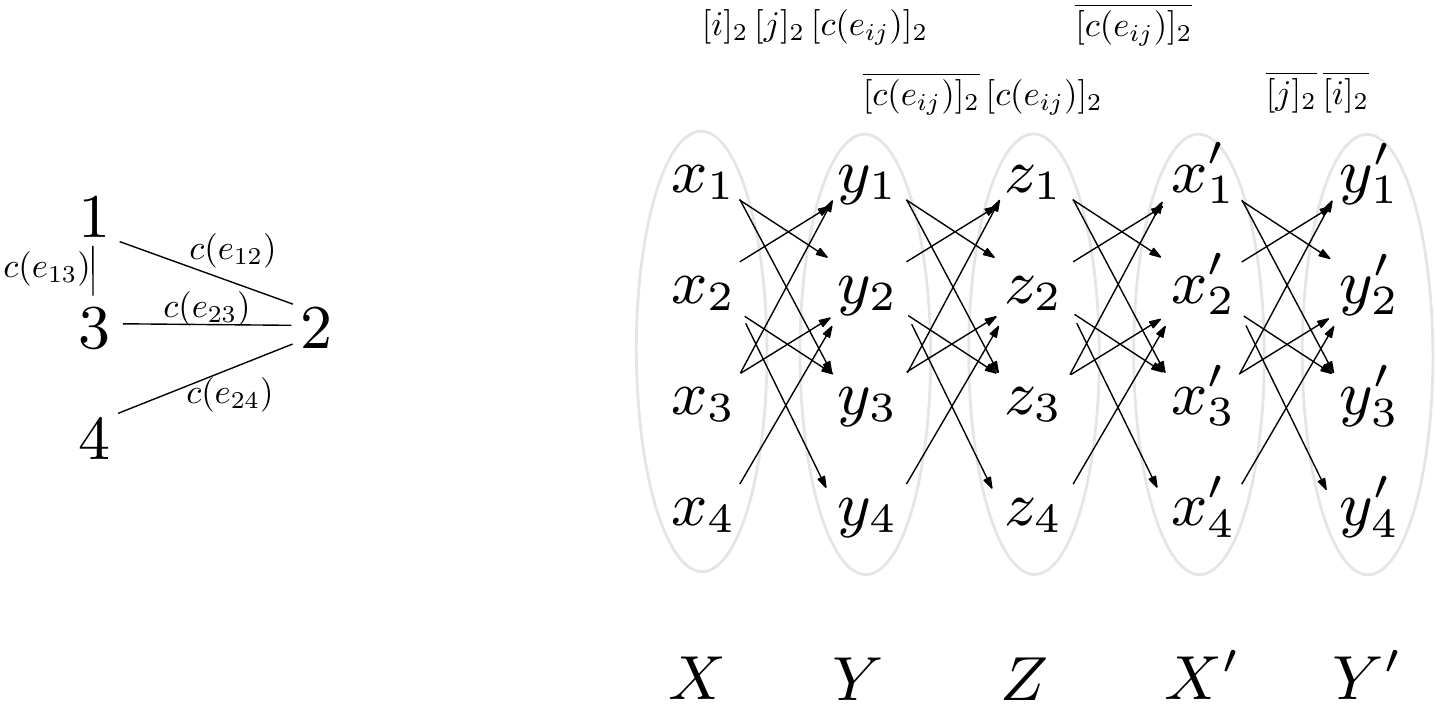}
	\end{center}
	
	\caption{The reduction from an instance of a AE-Mono$\Delta$ problem with $4$ vertices to the instance of all-pairs CFL reachability problem on subdivision-graph.}
	\label{fig:ae_mono_to_cflr}
	
\end{figure}

\begin{proof}
The reduction and its correctness is totally the same as in the reduction from (AE-Mono$\Delta$, $n^{2.5}$) to (sparse all-pairs PDS reachability with stack depth $b = 4 \lceil \log n \rceil$, $n^{1.25}$) presented in Sec.~\ref{subsec:pds}. 

Note that on every edge we get $4 \lceil \log n \rceil - 1$ additional vertices and thus the constructed graph is a $k$-subdivision-graph, $k = \Theta(\log n)$, with $5n = \Theta(n)$ ordinary vertices.
\end{proof}

Now we are ready to prove Theorem~\ref{lemma:no_expansion}.

\begin{proof}[Proof of Theorem~\ref{lemma:no_expansion}]
Suppose claim of Theorem~\ref{th:line-edges-1} works for $k$-subdivision-graphs, $k = \Omega(\log n)$. We want to apply it to the graph from the reduction in Lemma~\ref{lemma:no_exp_help} --- graph $D$. This is a $k$-subdivision-graph, $k = \Omega(\log n)$, moreover the corresponding graph $D'$ is an acyclic graph where the length of the paths between vertices do not exceed 4.

We have seen in Sec~\ref{sec:bounded_paths} that CFL reachability on DAGs can be solved in $\mathcal(n^{\omega})$ time. That means that at least one of $D_1, \ldots, D_k, k \ge 1$ should be not an acyclic graph. In other case AE-Mono$\Delta$ can be solved in $\tilde{\mathcal{O}}(n^{\omega})$ time by Theorem~\ref{th:cflr_on_dag}. Recall that $AE-Mono\Delta$ problem has $\mathcal{O}(n^{2.5})$ lower bound under 3SUM and APSP hypotheses. Together these two facts imply that both 3SUM and APSP hypotheses are false.

Moreover at least in one of $D_1, \ldots, D_k, k \ge 1$ the paths that we look for should be of length at least $k = \tilde{\Omega}(n^{c'}), c' \ge \frac{2.5}{\omega} - 1$ otherwise by Corollary~\ref{cor:paths} AE-Mono$\Delta$ can be solved in $\tilde{\mathcal{O}}(n^{2.5 - \epsilon}), \epsilon > 0$ time. 

In both of the cases in at least one of 3SUM and APSP hypotheses is true we need to considerably change the underlying graph $D'$ at least once before calling algorithm $A$ on it.  
\end{proof}

Note that the proven theorem also helps us to understand the structure of a graph that we should get in a reduction from AE-Mono$\Delta$ to CFL reachability if there exists one. Let $U$ be the graph on $n$ vertices in the instance of AE-Mono$\Delta$ problem. Then if at least one of APSP and 3SUM hypotheses is correct then the graph $D=(V, E, L)$ in CFL reachability problem should contain:

\begin{itemize}
    \item[-] cycle and the shortest $S$-path of length at least $\tilde{\Omega}(1)$ between some pair of vertices (at least $\tilde{\Omega}(n^{c'}), c' < \frac{2.5}{c \cdot \omega} - 1$ if $D$ has no more that $\mathcal{O}(n^c)$ vertices by Corollary~\ref{cor:paths})
    \item[-] or $n' = \Omega(n^{\frac{2.5}{\omega}})$ vertices (so that $\tilde{\mathcal{O}}(n'^{\omega}) \ge \tilde{\mathcal{O}}(n^{2.5})$)
\end{itemize} 

\section{Conclusion and future work}

In the course of this work, the following results were obtained:

\begin{enumerate}
    \item[1.] We have studied the existing algorithms for CFL reachability and proposed a faster algorithm for the subcase of finding paths of bounded length (see Sec.~\ref{subsec:basic} and~\ref{sec:bounded_paths}).
    \item[2.] We gave an overview on the reductions around CFL reachability problems (see Sec.~\ref{sec:map}).
    \item[3.1.] Some of the techniques used in existing conditional lower bounds were translated to CFL reachability and the problems connected to it, new conditional lower bounds were achieved (see Sec.~\ref{sec:lower} and Fig.~\ref{fig:map}). 
    \item[3.2.] We have presented new technique that may be used in creating new conditional lower bounds on CFL reachability and discussed its limitations (see Sec~\ref{sec:line_edges}).
    \item[3.3.] Several problems with lower bounds based on APSP, 3SUM and SAT hypotheses were studied. In Sec.~\ref{subsec:other_red} and~\ref{sec:line_edges} we discussed reasons why the reductions from them will be either hard, unlikely or will refute some hypotheses.
\end{enumerate}

There are still many open questions in the area of this work. We point out some of them. 

First question considers the existence of new conditional lower bounds for CFL reachability. We want to highlight two problems as a candidates to obtain them. They are APSP problem and OMV problem. We propose APSP problem and it's many subcubic equivalent analogues~\cite{10.1145/3186893} as APSP is connected to problems on paths, and OMV problem as it is connected to dynamic problems as is CFL reachability problem. 

One of the promising way of research is creating a reduction from all-pairs to s-t CFL reachability. There are known examples (e.g. All-Pairs Negative Triangle and Negative Triangle~\cite{Williams2009TriangleDV}) where such problems have the same complexity. Approach that was used in Triangle-connected problems is not naively applicable to CFL reachability as it is based on splitting sets of vertices in a tripartite graph. Still some existing ideas ideas may be taken into account when creating such a reduction between CFL reachability variations.

Despite the fact that Theorem~\ref{lemma:no_expansion} gives the intuition of the complexities behind the generalisation of the approach with line-edges it does admit the existence of other approaches: we can prove the analogous theorem for algorithms with some specific complexities, i.e. cubic or subcubic. It would be interesting to try another way to deal with long edges.

\printbibliography

@article{valiant1975general,
	title={General context-free recognition in less than cubic time},
	author={Valiant, Leslie G},
	journal={Journal of computer and system sciences},
	volume={10},
	number={2},
	pages={308--315},
	year={1975},
	publisher={Academic Press}
}

@article{10.1145/505241.505242,
	author = {Lee, Lillian},
	title = {Fast Context-Free Grammar Parsing Requires Fast Boolean Matrix Multiplication},
	year = {2002},
	issue_date = {January 2002},
	publisher = {Association for Computing Machinery},
	address = {New York, NY, USA},
	volume = {49},
	number = {1},
	issn = {0004-5411},
	url = {https://doi.org/10.1145/505241.505242},
	doi = {10.1145/505241.505242},
	journal = {J. ACM},
	month = jan,
	pages = {1–15},
	numpages = {15}
}

@article{abboud2018if,
	title={If the current clique algorithms are optimal, so is Valiant's parser},
	author={Abboud, Amir and Backurs, Arturs and Williams, Virginia Vassilevska},
	journal={SIAM Journal on Computing},
	volume={47},
	number={6},
	pages={2527--2555},
	year={2018},
	publisher={SIAM}
}

@article{10.1145/3186893,
	author = {Williams, Virginia Vassilevska and Williams, R. Ryan},
	title = {Subcubic Equivalences Between Path, Matrix, and Triangle Problems},
	year = {2018},
	issue_date = {September 2018},
	publisher = {Association for Computing Machinery},
	address = {New York, NY, USA},
	volume = {65},
	number = {5},
	issn = {0004-5411},
	url = {https://doi.org/10.1145/3186893},
	doi = {10.1145/3186893},
	journal = {J. ACM},
	month = aug,
	articleno = {27},
	numpages = {38}
}

@inproceedings{10.5555/646233.682379,
	author = {Ruzzo, Walter L.},
	title = {On the Complexity of General Context-Free Language Parsing and Recognition (Extended Abstract)},
	year = {1979},
	isbn = {3540095101},
	publisher = {Springer-Verlag},
	address = {Berlin, Heidelberg},
	booktitle = {Proceedings of the 6th Colloquium, on Automata, Languages and Programming},
	pages = {489–497},
	numpages = {9}
}

@article{10.1145/3158118,
	author = {Chatterjee, Krishnendu and Choudhary, Bhavya and Pavlogiannis, Andreas},
	title = {Optimal Dyck Reachability for Data-Dependence and Alias Analysis},
	year = {2017},
	issue_date = {January 2018},
	publisher = {Association for Computing Machinery},
	address = {New York, NY, USA},
	volume = {2},
	number = {POPL},
	url = {https://doi.org/10.1145/3158118},
	doi = {10.1145/3158118},
	journal = {Proc. ACM Program. Lang.},
	month = dec,
	articleno = {30},
	numpages = {30}
}

@inproceedings{bradford2017efficient,
	title={Efficient exact paths for Dyck and semi-Dyck labeled path reachability},
	author={Bradford, Phillip G},
	booktitle={2017 IEEE 8th Annual Ubiquitous Computing, Electronics and Mobile Communication Conference (UEMCON)},
	pages={247--253},
	year={2017},
	organization={IEEE}
}

@article{10.1145/3434315,
	author = {Mathiasen, Anders Alnor and Pavlogiannis, Andreas},
	title = {The Fine-Grained and Parallel Complexity of Andersen’s Pointer Analysis},
	year = {2021},
	issue_date = {January 2021},
	publisher = {Association for Computing Machinery},
	address = {New York, NY, USA},
	volume = {5},
	number = {POPL},
	url = {https://doi.org/10.1145/3434315},
	doi = {10.1145/3434315},
	journal = {Proc. ACM Program. Lang.},
	month = jan,
	articleno = {34},
	numpages = {29}
}

@inproceedings{10.5555/788019.788876,
	author = {Heintze, Nevin and McAllester, David},
	title = {On the Cubic Bottleneck in Subtyping and Flow Analysis},
	year = {1997},
	isbn = {0818679255},
	publisher = {IEEE Computer Society},
	address = {USA},
	booktitle = {Proceedings of the 12th Annual IEEE Symposium on Logic in Computer Science},
	pages = {342},
	series = {LICS '97}
}

@inproceedings{10.1145/2840728.2840746,
	author = {Carmosino, Marco L. and Gao, Jiawei and Impagliazzo, Russell and Mihajlin, Ivan and Paturi, Ramamohan and Schneider, Stefan},
	title = {Nondeterministic Extensions of the Strong Exponential Time Hypothesis and Consequences for Non-Reducibility},
	year = {2016},
	isbn = {9781450340571},
	publisher = {Association for Computing Machinery},
	address = {New York, NY, USA},
	url = {https://doi.org/10.1145/2840728.2840746},
	doi = {10.1145/2840728.2840746},
	booktitle = {Proceedings of the 2016 ACM Conference on Innovations in Theoretical Computer Science},
	pages = {261–270},
	numpages = {10},
	location = {Cambridge, Massachusetts, USA},
	series = {ITCS '16}
}

@article{Hanauer2020FasterFD,
	title={Faster Fully Dynamic Transitive Closure in Practice},
	author={Kathrin Hanauer and Monika Henzinger and Christian Schulz},
	journal={ArXiv},
	year={2020},
	volume={abs/2002.00813}
}

@article{10.1145/258994.259006,
	author = {Melski, David and Reps, Thomas},
	title = {Interconvertbility of Set Constraints and Context-Free Language Reachability},
	year = {1997},
	issue_date = {Dec. 1997},
	publisher = {Association for Computing Machinery},
	address = {New York, NY, USA},
	volume = {32},
	number = {12},
	issn = {0362-1340},
	url = {https://doi.org/10.1145/258994.259006},
	doi = {10.1145/258994.259006},
	journal = {SIGPLAN Not.},
	month = dec,
	pages = {74–89},
	numpages = {16}
}

@inproceedings{bringmann2019fine,
	title={Fine-Grained Complexity Theory},
	author={Bringmann, Karl},
	booktitle={36th International Symposium on Theoretical Aspects of Computer Science},
	pages={1},
	year={2019}
}

@article{10.1145/373243.360208,
	author = {Rehof, Jakob and F\"{a}hndrich, Manuel},
	title = {Type-Base Flow Analysis: From Polymorphic Subtyping to CFL-Reachability},
	year = {2001},
	issue_date = {March 2001},
	publisher = {Association for Computing Machinery},
	address = {New York, NY, USA},
	volume = {36},
	number = {3},
	issn = {0362-1340},
	url = {https://doi.org/10.1145/373243.360208},
	doi = {10.1145/373243.360208},
	journal = {SIGPLAN Not.},
	month = jan,
	pages = {54–66},
	numpages = {13}
}

@article{10.1145/1103845.1094817,
	author = {Sridharan, Manu and Gopan, Denis and Shan, Lexin and Bod\'{\i}k, Rastislav},
	title = {Demand-Driven Points-to Analysis for Java},
	year = {2005},
	issue_date = {October 2005},
	publisher = {Association for Computing Machinery},
	address = {New York, NY, USA},
	volume = {40},
	number = {10},
	issn = {0362-1340},
	url = {https://doi.org/10.1145/1103845.1094817},
	doi = {10.1145/1103845.1094817},
	journal = {SIGPLAN Not.},
	month = oct,
	pages = {59–76},
	numpages = {18}
}

@article{10.1145/1133255.1134027,
	author = {Sridharan, Manu and Bod\'{\i}k, Rastislav},
	title = {Refinement-Based Context-Sensitive Points-to Analysis for Java},
	year = {2006},
	issue_date = {June 2006},
	publisher = {Association for Computing Machinery},
	address = {New York, NY, USA},
	volume = {41},
	number = {6},
	issn = {0362-1340},
	url = {https://doi.org/10.1145/1133255.1134027},
	doi = {10.1145/1133255.1134027},
	journal = {SIGPLAN Not.},
	month = jun,
	pages = {387–400},
	numpages = {14}
}

@inproceedings{10.1145/298514.298576,
	author = {Yannakakis, Mihalis},
	title = {Graph-Theoretic Methods in Database Theory},
	year = {1990},
	isbn = {0897913523},
	publisher = {Association for Computing Machinery},
	address = {New York, NY, USA},
	url = {https://doi.org/10.1145/298514.298576},
	doi = {10.1145/298514.298576},
	booktitle = {Proceedings of the Ninth ACM SIGACT-SIGMOD-SIGART Symposium on Principles of Database Systems},
	pages = {230–242},
	numpages = {13},
	location = {Nashville, Tennessee, USA},
	series = {PODS '90}
}

@inproceedings{10.1145/1328438.1328460,
	author = {Chaudhuri, Swarat},
	title = {Subcubic Algorithms for Recursive State Machines},
	year = {2008},
	isbn = {9781595936899},
	publisher = {Association for Computing Machinery},
	address = {New York, NY, USA},
	url = {https://doi.org/10.1145/1328438.1328460},
	doi = {10.1145/1328438.1328460},
	booktitle = {Proceedings of the 35th Annual ACM SIGPLAN-SIGACT Symposium on Principles of Programming Languages},
	pages = {159–169},
	numpages = {11},
	location = {San Francisco, California, USA},
	series = {POPL '08}
}

@inproceedings{10.1145/199448.199462,
	author = {Reps, Thomas and Horwitz, Susan and Sagiv, Mooly},
	title = {Precise Interprocedural Dataflow Analysis via Graph Reachability},
	year = {1995},
	isbn = {0897916921},
	publisher = {Association for Computing Machinery},
	address = {New York, NY, USA},
	url = {https://doi.org/10.1145/199448.199462},
	doi = {10.1145/199448.199462},
	booktitle = {Proceedings of the 22nd ACM SIGPLAN-SIGACT Symposium on Principles of Programming Languages},
	pages = {49–61},
	numpages = {13},
	location = {San Francisco, California, USA},
	series = {POPL '95}
}

@inproceedings{10.1145/2746539.2746609,
	author = {Henzinger, Monika and Krinninger, Sebastian and Nanongkai, Danupon and Saranurak, Thatchaphol},
	title = {Unifying and Strengthening Hardness for Dynamic Problems via the Online Matrix-Vector Multiplication Conjecture},
	year = {2015},
	isbn = {9781450335362},
	publisher = {Association for Computing Machinery},
	address = {New York, NY, USA},
	url = {https://doi.org/10.1145/2746539.2746609},
	doi = {10.1145/2746539.2746609},
	booktitle = {Proceedings of the Forty-Seventh Annual ACM Symposium on Theory of Computing},
	pages = {21–30},
	numpages = {10},
	location = {Portland, Oregon, USA},
	series = {STOC '15}
}

@misc{shemetova2021algorithm,
	title={One Algorithm to Evaluate Them All: Unified Linear Algebra Based Approach to Evaluate Both Regular and Context-Free Path Queries}, 
	author={Ekaterina Shemetova and Rustam Azimov and Egor Orachev and Ilya Epelbaum and Semyon Grigorev},
	year={2021},
	eprint={2103.14688},
	archivePrefix={arXiv},
	primaryClass={cs.DB}
}

@INPROCEEDINGS{8948597,  
	author={van den Brand, Jan and Nanongkai, Danupon and Saranurak, Thatchaphol},  
	booktitle={2019 IEEE 60th Annual Symposium on Foundations of Computer Science (FOCS)},   
	title={Dynamic Matrix Inverse: Improved Algorithms and Matching Conditional Lower Bounds},   
	year={2019},  
	volume={},  
	number={},  
	pages={456-480},  
	doi={10.1109/FOCS.2019.00036}}

@article{REPS1998701,
	title = {Program analysis via graph reachability1An abbreviated version of this paper appeared as an invited paper in the Proceedings of the 1997 International Symposium on Logic Programming [84].1},
	journal = {Information and Software Technology},
	volume = {40},
	number = {11},
	pages = {701-726},
	year = {1998},
	issn = {0950-5849},
	doi = {https://doi.org/10.1016/S0950-5849(98)00093-7},
	url = {https://www.sciencedirect.com/science/article/pii/S0950584998000937},
	author = {Thomas Reps}
}

@article { SubgraphQueriesbyContextfreeGrammars,
      author = "Petteri Sevon and Lauri Eronen",
      title = "Subgraph Queries by Context-free Grammars",
      journal = "Journal of Integrative Bioinformatics",
      year = "2008",
      publisher = "De Gruyter",
      address = "Berlin, Boston",
      volume = "5",
      number = "2",
      doi = "https://doi.org/10.1515/jib-2008-100",
      pages=      "157 -- 172",
      url = "https://www.degruyter.com/view/journals/jib/5/2/article-p157.xml"
}

@article{hansen2021tight,
  title={Tight bounds for reachability problems on one-counter and pushdown systems},
  author={Hansen, Jakob Cetti and Kjelstr{\o}m, Adam Husted and Pavlogiannis, Andreas},
  journal={Information Processing Letters},
  volume={171},
  pages={106135},
  year={2021},
  publisher={Elsevier}
}

@book{10.5555/524279,
author = {Sipser, Michael},
title = {Introduction to the Theory of Computation},
year = {1996},
isbn = {053494728X},
publisher = {International Thomson Publishing},
edition = {1st}
}

@book{10.5555/1196416,
author = {Hopcroft, John E. and Motwani, Rajeev and Ullman, Jeffrey D.},
title = {Introduction to Automata Theory, Languages, and Computation (3rd Edition)},
year = {2006},
isbn = {0321455363},
publisher = {Addison-Wesley Longman Publishing Co., Inc.},
address = {USA}
}

@article{schepper2018complexity,
  title={The Complexity of Formal Language Decision Problems},
  author={Schepper, Philipp Johann},
  year={2018}
}

@article{10.1145/3571252,
author = {Koutris, Paraschos and Deep, Shaleen},
title = {The Fine-Grained Complexity of CFL Reachability},
year = {2023},
issue_date = {January 2023},
publisher = {Association for Computing Machinery},
address = {New York, NY, USA},
volume = {7},
number = {POPL},
url = {https://doi.org/10.1145/3571252},
doi = {10.1145/3571252},
journal = {Proc. ACM Program. Lang.},
month = {jan},
articleno = {59},
numpages = {27}
}

@article{ch-sch,
author = {Matthews, G.},
year = {2014},
month = {10},
pages = {388-389},
title = {Chomsky N. and Schützenberger M. P.. The algebraic theory of context-free languages. Computer programming and formal systems, edited by Braffort P. and Hirschberg D., Studies in logic and the foundations of mathematics, North-Holland Publishing Company, Amsterdam 1963, pp. 118–161.},
volume = {32},
journal = {The Journal of Symbolic Logic},
doi = {10.2307/2270782}
}

@inproceedings{williams2018some,
  title={On some fine-grained questions in algorithms and complexity},
  author={Williams, Virginia Vassilevska},
  booktitle={Proceedings of the international congress of mathematicians: Rio de janeiro 2018},
  pages={3447--3487},
  year={2018},
  organization={World Scientific}
}

@article{10.1006/jcss.2000.1727,
author = {Impagliazzo, Russell and Paturi, Ramamohan},
title = {On the Complexity of K-SAT},
year = {2001},
issue_date = {March 2001},
publisher = {Academic Press, Inc.},
address = {USA},
volume = {62},
number = {2},
issn = {0022-0000},
url = {https://doi.org/10.1006/jcss.2000.1727},
doi = {10.1006/jcss.2000.1727},
journal = {J. Comput. Syst. Sci.},
month = {mar},
pages = {367–375},
numpages = {9}
}

@inproceedings{10.5555/3458064.3458096,
author = {Alman, Josh and Williams, Virginia Vassilevska},
title = {A Refined Laser Method and Faster Matrix Multiplication},
year = {2021},
isbn = {9781611976465},
publisher = {Society for Industrial and Applied Mathematics},
address = {USA},
booktitle = {Proceedings of the Thirty-Second Annual ACM-SIAM Symposium on Discrete Algorithms},
pages = {522–539},
numpages = {18},
location = {Virtual Event, Virginia},
series = {SODA '21}
}

@inproceedings{10.1007/11786986_24,
author = {Vassilevska, Virginia and Williams, Ryan and Yuster, Raphael},
title = {Finding the Smallest H-Subgraph in Real Weighted Graphs and Related Problems},
year = {2006},
isbn = {3540359044},
publisher = {Springer-Verlag},
address = {Berlin, Heidelberg},
url = {https://doi.org/10.1007/11786986_24},
doi = {10.1007/11786986_24},
booktitle = {Proceedings of the 33rd International Conference on Automata, Languages and Programming - Volume Part I},
pages = {262–273},
numpages = {12},
location = {Venice, Italy},
series = {ICALP'06}
}

@article{10.1145/1798596.1798597,
author = {Vassilevska, Virginia and Williams, Ryan and Yuster, Raphael},
title = {Finding Heaviest H-Subgraphs in Real Weighted Graphs, with Applications},
year = {2010},
issue_date = {June 2010},
publisher = {Association for Computing Machinery},
address = {New York, NY, USA},
volume = {6},
number = {3},
issn = {1549-6325},
url = {https://doi.org/10.1145/1798596.1798597},
doi = {10.1145/1798596.1798597},
journal = {ACM Trans. Algorithms},
month = {jul},
articleno = {44},
numpages = {23}
}

@article{https://doi.org/10.4230/lipics.itcs.2020.53,
doi = {10.4230/LIPICS.ITCS.2020.53},
url = {https://drops.dagstuhl.de/opus/volltexte/2020/11738/},
author = {Lincoln, Andrea and Polak, Adam and Vassilevska Williams, Virginia},
language = {en},
title = {Monochromatic Triangles, Intermediate Matrix Products, and Convolutions},
publisher = {Schloss Dagstuhl - Leibniz-Zentrum fuer Informatik GmbH, Wadern/Saarbruecken, Germany},
year = {2020},
copyright = {Creative Commons Attribution 3.0 Unported license (CC-BY 3.0)}
}

@inproceedings{10.1145/2746539.2746594,
author = {Abboud, Amir and Vassilevska Williams, Virginia and Yu, Huacheng},
title = {Matching Triangles and Basing Hardness on an Extremely Popular Conjecture},
year = {2015},
isbn = {9781450335362},
publisher = {Association for Computing Machinery},
address = {New York, NY, USA},
url = {https://doi.org/10.1145/2746539.2746594},
doi = {10.1145/2746539.2746594},
booktitle = {Proceedings of the Forty-Seventh Annual ACM Symposium on Theory of Computing},
pages = {41–50},
numpages = {10},
location = {Portland, Oregon, USA},
series = {STOC '15}
}

@article{10.1016/j.tcs.2005.09.023,
author = {Williams, Ryan},
title = {A New Algorithm for Optimal 2-Constraint Satisfaction and Its Implications},
year = {2005},
issue_date = {8 December 2005},
publisher = {Elsevier Science Publishers Ltd.},
address = {GBR},
volume = {348},
number = {2},
issn = {0304-3975},
url = {https://doi.org/10.1016/j.tcs.2005.09.023},
doi = {10.1016/j.tcs.2005.09.023},
journal = {Theor. Comput. Sci.},
month = {dec},
pages = {357–365},
numpages = {9}
}

@article{CHOMSKY1959137,
title = {On certain formal properties of grammars},
journal = {Information and Control},
volume = {2},
number = {2},
pages = {137-167},
year = {1959},
issn = {0019-9958},
doi = {https://doi.org/10.1016/S0019-9958(59)90362-6},
url = {https://www.sciencedirect.com/science/article/pii/S0019995859903626},
author = {Noam Chomsky}
}

@inproceedings{Williams2009TriangleDV,
  title={Triangle Detection Versus Matrix Multiplication : A Study of Truly Subcubic Reducibility*},
  author={Virginia Vassilevska Williams and Ryan Williams},
  year={2009}
}

@article{10.1137/0201022,
author = {Aho, Alfred V. and Peterson, Thomas G.},
title = {A Minimum Distance Error-Correcting Parser for Context-Free Languages},
year = {1972},
issue_date = {Dec 1972},
publisher = {Society for Industrial and Applied Mathematics},
address = {USA},
volume = {1},
number = {4},
issn = {0097-5397},
url = {https://doi.org/10.1137/0201022},
doi = {10.1137/0201022},
journal = {SIAM J. Comput.},
month = {dec},
pages = {305–312},
numpages = {8}
}

@article{DBLP:journals/corr/Saha14,
  author       = {Barna Saha},
  title        = {Faster Language Edit Distance, Connection to All-pairs Shortest Paths
                  and Related Problems},
  journal      = {CoRR},
  volume       = {abs/1411.7315},
  year         = {2014},
  url          = {http://arxiv.org/abs/1411.7315},
  eprinttype    = {arXiv},
  eprint       = {1411.7315},
  bibsource    = {dblp computer science bibliography, https://dblp.org}
}

@article{DBLP:journals/corr/BringmannGSW17,
  author       = {Karl Bringmann and
                  Fabrizio Grandoni and
                  Barna Saha and
                  Virginia Vassilevska Williams},
  title        = {Truly Sub-cubic Algorithms for Language Edit Distance and {RNA} Folding
                  via Fast Bounded-Difference Min-Plus Product},
  journal      = {CoRR},
  volume       = {abs/1707.05095},
  year         = {2017},
  url          = {http://arxiv.org/abs/1707.05095},
  eprinttype    = {arXiv},
  eprint       = {1707.05095},
  bibsource    = {dblp computer science bibliography, https://dblp.org}
}

@article{PIERRE1992279,
title = {Rational indexes of generators of the cone of context-free languages},
journal = {Theoretical Computer Science},
volume = {95},
number = {2},
pages = {279-305},
year = {1992},
issn = {0304-3975},
doi = {https://doi.org/10.1016/0304-3975(92)90269-L},
url = {https://www.sciencedirect.com/science/article/pii/030439759290269L},
author = {Laurent Pierre}
}

@article{10.1006/jcss.1997.1388,
author = {Alon, Noga and Galil, Zvi and Margalit, Oded},
title = {On the Exponent of the All Pairs Shortest Path Problem},
year = {1997},
issue_date = {April 1997},
publisher = {Academic Press, Inc.},
address = {USA},
volume = {54},
number = {2},
issn = {0022-0000},
url = {https://doi.org/10.1006/jcss.1997.1388},
doi = {10.1006/jcss.1997.1388},
journal = {J. Comput. Syst. Sci.},
month = {apr},
pages = {255–262},
numpages = {8}
}

@article{lspidwnapew,
author = {Bradford, Phillip and Thomas, David},
year = {2009},
month = {07},
pages = {},
title = {Labeled shortest paths in digraphs with negative and positive edge weights},
volume = {43},
journal = {http://dx.doi.org/10.1051/ita/2009011},
doi = {10.1051/ita/2009011}
}

@inproceedings{Rytter2005FastRO,
  title={Fast Recognition of Pushdown Automaton Context-free Languages},
  author={Wojciech Rytter},
  year={2005}
}

@article{10.1145/3498702,
author = {Chistikov, Dmitry and Majumdar, Rupak and Schepper, Philipp},
title = {Subcubic Certificates for CFL Reachability},
year = {2022},
issue_date = {January 2022},
publisher = {Association for Computing Machinery},
address = {New York, NY, USA},
volume = {6},
number = {POPL},
url = {https://doi.org/10.1145/3498702},
doi = {10.1145/3498702},
journal = {Proc. ACM Program. Lang.},
month = {jan},
articleno = {41},
numpages = {29}
}
\end{document}